%% file: biv2.tex
\documentclass[runningheads]{llncs}
\makeatletter
\let\c@proposition\c@theorem
\let\c@corollary\c@theorem
\let\c@lemma\c@theorem
\let\c@definition\c@theorem
\let\c@example\c@theorem
\let\c@remark\c@theorem
\def\@listI{\leftmargin\leftmargini
            \parsep 0\p@ \@plus1\p@ \@minus\p@
            \topsep 1\p@ \@plus1\p@ \@minus\p@
            \itemsep0\p@}
\let\@listi\@listI
\@listi
\makeatother

\usepackage[T1]{fontenc}
\usepackage{lmodern}
\usepackage{graphicx}
\usepackage{adjustbox}
\usepackage{proofzilla}
\usepackage{virginialake}
\vlnostructuressyntax
\vlnosmallleftlabels
\odframefalse
\odbackgroundtrue
\odbackgroundfirstfalse
\odframefirstfalse

\usepackage{amsfonts}
\usepackage{amsmath}
\usepackage{mathtools}
\usepackage{hyperref}
\usepackage{cleveref}
\usepackage{thm-restate}
\usepackage{xspace}

\usepackage{stmaryrd}
\usepackage{arydshln}

\input{macros.tex}

\def\cneg#1{{#1}^\lbot}
\def\cnegneg#1{{#1}^{\lbot\lbot}}

\def\wsymb{\circ}
\def\bsymb{\bullet}

\def\wrule{\!\!\circ}
\def\brule{\!\!\bullet}
\def\isw#1{{#1}^\circ}
\def\isb#1{{#1}^\bullet}

\def\wA{\isw A}
\def\wB{\isw B}
\def\bA{\isb A}
\def\bB{\isb B}

\def\wC{\isw C}

\renewcommand{\lunit}{\mathbb{I}}

\newcommand{\lneg}{^\bot}
\newcommand{\emb}[1]{#1^{\flat}}

\newcommand{\odpr}[2][]{\od{\odp{#1}{#2}{}}}
\newcommand{\oddr}[3][]{\od{\odd{\odh{#2}}{#1}{#3}{}}}
\newcommand{\oddb}[3][]{\od{\odd{\odh{#2}}{#1}{#3}{\hskip-.5em\bullet}}}
\newcommand{\oddrS}[4][]{\od{\odd{\odh{#2}}{#1}{#3}{#4}}}
\newcommand{\oddbS}[4][]{\od{\odd{\odh{#2}}{#1}{#3}{\hskip-.5em\bullet \; #4}}}
\newcommand{\odprP}[2][]{\left(\od{\odp{#1}{#2}{}}\right)}
\newcommand{\oddrP}[3][]{\left(\od{\odd{\odh{#2}}{#1}{#3}{}}\right)}
\newcommand{\oddbP}[3][]{\left(\od{\odd{\odh{#2}}{#1}{#3}{\hskip-.5em\bullet}}\right)}
\newcommand{\odtP}[4]{\left(\odt{#1}{#2}{#3}{#4}\right)}
\newcommand{\odnP}[4]{\left(\odn{#1}{#2}{#3}{#4}\right)}

\def\acutl{\mathsf{a\mhyphen cut_L}}
\def\acutr{\mathsf{a\mhyphen cut_R}}

\begin{document}
%%%%%%%%%%%%%%%%%%%%%%%%%%%%%%%%%%%%%%%%%%%%%%%%%%%%%%%%%%%%%%%%
%%%%%%%%%%%%%%%%%%%%%%%%%%%%%%%%%%%%%%%%%%%%%%%%%%%%%%%%%%%%%%%%
%%%%%%%%%%%%%%%%%%%%%%%%%%%%%%%%%%%%%%%%%%%%%%%%%%%%%%%%%%%%%%%%
%%%%%%%%%%%%%%%%%%%%%%%%%%%%%%%%%%%%%%%%%%%%%%%%%%%%%%%%%%%%%%%%

\title{Intuitionistic BV (Extended Version)\thanks{Partially supported by the French Ministry for Europe and Foreign Affairs (MEAE), the Embassy of France in the UK, and the French Ministry of Higher Education and Research (MESR), via the PHC Sophie Germain project
``Using Formal Logic to Reduce Bias in Large Language Models''}}
\titlerunning{Intuitionistic BV}
% If the paper title is too long for the running head, you can set an abbreviated paper title here

\author{
  Matteo Acclavio\inst{1}\orcidID{0000-0002-0425-2825} 
  \and
  \\
  Lutz Stra\ss burger\inst{2}\orcidID{0000-0003-4661-6540}
}
\authorrunning{M.~Acclavio \and L.~Stra\ss burger}
\institute{
  Sussex University, Brighton, UK
\and
INRIA Saclay,
Palaiseau, France
}

\maketitle

\begin{abstract}
  We present the logic IBV, which is an intuitionistic version of BV, in the sense that its restriction to the MLL connectives is exactly IMLL, the intuitionistic version of MLL. For this logic we give a deep inference proof system and show cut elimination. We also show that the logic obtained from IBV by dropping the associativity of the new non-commutative seq-connective is an intuitionistic variant of the recently introduced logic NML. For this logic, called INML, we give a cut-free sequent calculus.  
\end{abstract}

%%%%%%%%%%%%%%%%%%%%%%%%%%%%%%%%%%%%%%%%%%%%%%%%%%%%%%%%%%%%%%%%
%%%%%%%%%%%%%%%%%%%%%%%%%%%%%%%%%%%%%%%%%%%%%%%%%%%%%%%%%%%%%%%%
%%%%%%%%%%%%%%%%%%%%%%%%%%%%%%%%%%%%%%%%%%%%%%%%%%%%%%%%%%%%%%%%
\section{Introduction}
%%%%%%%%%%%%%%%%%%%%%%%%%%%%%%%%%%%%%%%%%%%%%%%%%%%%%%%%%%%%%%%%
%%%%%%%%%%%%%%%%%%%%%%%%%%%%%%%%%%%%%%%%%%%%%%%%%%%%%%%%%%%%%%%%
%%%%%%%%%%%%%%%%%%%%%%%%%%%%%%%%%%%%%%%%%%%%%%%%%%%%%%%%%%%%%%%%

The logic $\BV$ is a conservative extension of multiplicative linear logic with mix ($\MLLx$) with a self-dual non-commutative connective ($\lseq$) called \emph{seq}.
It was introduced by Guglielmi in \cite{gug:SIS,gug:str:01} in the attempt of providing cut-free deduction system for Retoré's $\pomset$ logic \cite{retore:phd,ret:newPomset}.\footnote{
  The inclusion of $\BV$ in $\pomset$ has been known since the introduction of $\BV$~\cite{lutz:phd}. However, that this inclusion is strict has only been proven recently \cite{tito:lutz:csl22,tito:str:SIS-III}.
}
To this end, Guglielmi developed the deep inference formalism of the calculus of structures to deal with the limitations of the traditional proof systems based on Gentzen's work (sequent calculi and natural deduction).
In fact, as shown by Tui in \cite{tiu:SIS-II}, the presence of the non-commutative connective $\lseq$ makes it impossible to define a cut-free sequent calculus for $\BV$.
Nonetheless, the logic $\BV$ has found applications in the formalisation of process algebras (e.g., \cite{bru:02} and \cite{hor:tiu:tow,hor:nom}), in typing linear lambda calculus with explicit substitutions~\cite{roversi:lambda}, and in describing quantum computations through \emph{BV-categories}~\cite{blu:pan:slav:deep,sim:kiss:BV}.

In this paper, we want to define an intuitionistic variant of $\BV$. The biggest obstacle towards that goal is the fact that in $\BV$ the unit $\lunit$  is \emph{self-dual} and is shared not only by the connectives $\ltens$ and $\lpar$, but also the new \emph{self-dual} connective~$\lseq$.
It is well-known that intuitionistic multiplicative linear logic ($\MiLL$) can be obtained by polarizing formulas in multiplicative linear logic ($\MLL$).
But the self-duality of the unit $\lunit$ and the seq-connective $\lseq$ makes it difficult to extend this polarization to $\BV$.
Our solution to this problem is to make the unit  $\lunit$ only `half a unit'. That means that we no longer have $A\lseq \lunit\equiv A\equiv\lunit\lseq A$, but only $A\limp A\lseq \lunit$ and $A\limp \lunit\lseq A$.

In classical $\BV$ the triple $\tuple{\ltens,\lpar,\lunit}$ forms an \emph{isomix category} \cite{cockett:seely:97}, and the connective $\lseq$ is a \emph{degenerate linear functor} (in the sense of \cite{blu:pan:slav:deep}), that is, it validates the following implication.\footnote{The fact that the unit $\lunit$ of the $\ltens$ and $\lpar$ is also the (left and right) unit for $\lseq$ is a consequence of the definition of degenerate linear functor \cite{blu:pan:slav:deep}.}
\begin{equation}\label{eq:degfun}
  ((A\lseq B)\ltens (C\lseq D))\limp ((A\ltens C)\lseq (B\ltens D))
\end{equation}
We define \emph{intuitionistic $\BV$} ($\BiV$) by extending intuitionistic multiplicative linear logic ($\MiLL$), where the triple $\tuple{\ltens,\limp,\lunit}$ forms a symmetric monoidal closed structure, with a non-commutative connective $\lseq$ validating \Cref{eq:degfun} and the unit laws $A\limp (\lunit\lseq A)$ and $A\limp(A\lseq \lunit)$.%
\footnote{\label{fn:collapse}
  If $(\lunit \lseq A) \limp A$ and $(A \lseq \lunit) \limp A$ were valid in $\IBV$, then the connectives $\ltens$ and $\lseq$ would collapse. 
  See \Cref{prop:collapse} in \Cref{app:collapse} for details.
}

%%%%%%%%%%%%%%%%%%%%%%%%%%%%%%%%%%%%%%%%%%%%%%%%%%%%%%%%%%%%%%%%
%%%%%%%%%%%%%%%%%%%%%%%%%%%%%%%%%%%%%%%%%%%%%%%%%%%%%%%%%%%%%%%%
\begin{figure}[t]
  $$
    \begin{array}{c@{\qquad}c@{\quad\qquad}c}
      \vpz{2m}{\MLLx}& \vpz4{\NML} & \vpz{6}{\BV}
      \\\\
      % \\\\
      \vpz1{\MiLL} & \vmod3{\NiML} & \vmod5{\BiV}
    \end{array}
    \Dedges{pz1/mod3,mod3/mod5,pz2m/pz4,pz4/pz6}
    \Dedges{pz1/pz2m,mod3/pz4,mod5/pz6}
  $$
  \caption{
    Proof systems discussed in this paper. The ones in the boxes are new.
  }
  \label{fig:systems}
\end{figure}
%%%%%%%%%%%%%%%%%%%%%%%%%%%%%%%%%%%%%%%%%%%%%%%%%%%%%%%%%%%%%%%%
%%%%%%%%%%%%%%%%%%%%%%%%%%%%%%%%%%%%%%%%%%%%%%%%%%%%%%%%%%%%%%%%

\paragraph{Contributions}
We give a deep inference proof system for $\BiV$ (Sections~\ref{sec:ibv} and~\ref{sec:properties}), and we prove cut elimination via a \emph{splitting lemma} (in Section~\ref{sec:splitting}). We argue that $\BiV$ is indeed the intuitionistic version of $\BV$, by showing 
that (i) $\BiV$ is a conservative extension of $\IMLL$ (Section~\ref{sec:IMLLtoIBV}),
and that
(ii)  the unit-free version of $\BV$ is a conservative extension of the unit-free version of $\BiV$ (Section~\ref{sec:IBVtoBV}).\footnote{
  Note that $\BV$ is not conservative over $\MLL$, because it is conservative over $\MLLx$, and $\MLLx$ is not conservative over $\MLL$.
}

Finally, in Section~\ref{sec:assoc}, we present a weaker logic, called $\INML$, in which the connective $\lseq$ is not associative. We give a cut-free sequent calculus for $\NiML$, which is obtained by considering a single-conclusion two-sided version of the sequent calculus for the \emph{non-commutative multiplicative logic} ($\NML$) from \cite{acc:man:PN,acc:man:mon:ESOP}.
We prove that $\NiML$ is another conservative extension of $\MiLL$, which can be conservatively extended to $\NML$ and to $\BiV$ (see \Cref{fig:systems}).

%%%%%%%%%%%%%%%%%%%%%%%%%%%%%%%%%%%%%%%%%%%%%%%%%%%%%%%%%%%%%%%%
%%%%%%%%%%%%%%%%%%%%%%%%%%%%%%%%%%%%%%%%%%%%%%%%%%%%%%%%%%%%%%%%
%%%%%%%%%%%%%%%%%%%%%%%%%%%%%%%%%%%%%%%%%%%%%%%%%%%%%%%%%%%%%%%%
\section{Formulas and Inference Rules}\label{sec:ibv}
%%%%%%%%%%%%%%%%%%%%%%%%%%%%%%%%%%%%%%%%%%%%%%%%%%%%%%%%%%%%%%%%
%%%%%%%%%%%%%%%%%%%%%%%%%%%%%%%%%%%%%%%%%%%%%%%%%%%%%%%%%%%%%%%%
%%%%%%%%%%%%%%%%%%%%%%%%%%%%%%%%%%%%%%%%%%%%%%%%%%%%%%%%%%%%%%%%

%%%%%%%%%%%%%%%%%%%%%%%%%%%%%%%%%%%%%%%%%%%%%%%%%%%%%%%%%%%%%%%%
% Rules IBV
%%%%%%%%%%%%%%%%%%%%%%%%%%%%%%%%%%%%%%%%%%%%%%%%%%%%%%%%%%%%%%%%
\begin{figure}[t]
  \def\myskip{-1.5ex}
  $$
  \begin{array}{c}
    \nodn{\lunit}{\waidr}{a \limp a}{\wrule}
    \qquad
    \nodn{A}{\wsudr}{\lunit \lseq A}{\wrule}
    \qquad
    \nodn{A}{\wcudr}{A \lseq \lunit}{\wrule}
    \qquad
    \nodn{A \ltens B}{\wrr}{A \lseq B}{\wrule}
    \qquad
    \nodn{A \lseq B}{\brr}{A \ltens B}{\brule}
    \\ \\[\myskip]
    \nodn{A \ltens (B\limp C)}{\wlsr}{(A\limp B) \limp C}{\wrule}
    \qquad
    \nodn{(A\limp B)\ltens C}{\wtsr}{A\limp(B\ltens C)}{\wrule}
    \qquad
    \nodn{(A\limp B) \limp C}{\blsr}{A \ltens (B\limp C)}{\brule}
    \qquad
    \nodn{A\limp(B\ltens C)}{\btsr}{(A\limp B)\ltens C}{\brule}
    \\ \\[\myskip]
    \nodn{(A\limp B)\lseq C}{\wrsq}{A\limp(B\lseq C)}{\wrule}
    \qquad
    \nodn{B \lseq (A\limp C)}{\wlsq}{A\limp(B\lseq C)}{\wrule}
    \qquad
    \nodn{(A\ltens B)\lseq C}{\brsq}{A\ltens(B\lseq C)}{\brule}
    \qquad
    \nodn{B\lseq (A\ltens C)}{\blsq}{A\ltens(B\lseq C)}{\brule}
    \\ \\[\myskip]
    \nodn{(A\limp B)\lseq(C\limp D)}{\wqdr}{(A\lseq C)\limp (B\lseq D)}{\wrule}
    \qquad\qquad\qquad
    \nodn{(A\ltens B)\lseq(C\ltens D)}{\bqdr}{(A\lseq C)\ltens (B\lseq D)}{\brule}
    \\\\[\myskip]
    \hdashline
    \\[\myskip]
    \nodt{A\ltens B}{\tcomr}{B\ltens A}{}
  \qquad
    \nodt{(A\ltens B)\ltens C}{\tassr}{A\ltens (B\ltens C)}{}
  \qquad
    \nodt{(A\lseq B)\lseq C}{\sassl}{A\lseq (B\lseq C)}{}
  \qquad
    \nodt{A\lseq (B\lseq C)}{\sassr}{(A\lseq B)\lseq C}{}
  \\\\[\myskip]
  \nodt{A}{\tudr}{\lunit \ltens A}{}
  \qquad
  \nodt{A}{\ludr}{\lunit \limp A}{}
  \qquad
  \nodt{(A\ltens B)\limp C}{\curry}{A\limp (B\limp C)}{}
  \qquad
  \nodt{A\limp (B\limp C)}{\uncurry}{(A\ltens B)\limp C}{}
  \end{array}
  $$
  \vadjust{\vskip-2ex}
  \vadjust{\vskip-2ex}
  \caption{
    Inference rules for system $\IBV$
  }
  \label{fig:IBV}
\end{figure}
%%%%%%%%%%%%%%%%%%%%%%%%%%%%%%%%%%%%%%%%%%%%%%%%%%%%%%%%%%%%%%%%
%%%%%%%%%%%%%%%%%%%%%%%%%%%%%%%%%%%%%%%%%%%%%%%%%%%%%%%%%%%%%%%%

%%%%%%%%%%%%%%%%%%%%%%%%%%%%%%%%%%%%%%%%%%%%%%%%%%%%%%%%%%%%%%%%
% Rules SBiV
%%%%%%%%%%%%%%%%%%%%%%%%%%%%%%%%%%%%%%%%%%%%%%%%%%%%%%%%%%%%%%%%
\begin{figure}[t]
  \def\myskip{-1.5ex}
  $$
  \begin{array}{c}
    \nodn{a\limp a}{\baiur}{\lunit}{\brule}
    \qquad
    \nodn{\lunit\lseq A}{\bsudr}{A}{\brule}
    \qquad
    \nodn{A\lseq\lunit}{\bcudr}{A}{\brule}
    \qquad
    \nodt{\lunit\ltens A}{\tuur}{A}{}
    \qquad
    \nodt{\lunit\limp A}{\luur}{A}{}
  \\\\[\myskip]    
    \nodn{(A\lseq C)\ltens(B\lseq D)}{\wqur}{(A\ltens B)\lseq (C\ltens D)}{\wrule}
    \qquad
    \nodn{(A\lseq C)\limp (B\lseq D)}{\bqur}{(A\limp B)\lseq (C\limp D)}{\brule}
  \end{array}
  $$
  \caption{Additional rules for $\SBiV$.}
  \label{fig:SIBV}
\end{figure}
%%%%%%%%%%%%%%%%%%%%%%%%%%%%%%%%%%%%%%%%%%%%%%%%%%%%%%%%%%%%%%%%
%%%%%%%%%%%%%%%%%%%%%%%%%%%%%%%%%%%%%%%%%%%%%%%%%%%%%%%%%%%%%%%%

We consider \defn{formulas} generated from a countable set $\atomset=\set{a,b,c,\ldots}$ of atoms, a unit~$\lunit$, and the binary connectives implication~$\limp$, tensor~$\ltens$, and seq~$\lseq$~:
\begin{equation}\label{eq:formulas}
  A,B \coloneq a \mid \lunit \mid   A\ltens B\mid A\limp B\mid A\lseq B
  \qquad\qquad
  a \in \atomset
\end{equation}
A formula is \defn{unit-free} if it contains no occurrences of $\lunit$.
In order to define the deep inference rules for our systems, we need to define \defn{contexts}, which are formulas where one atom occurrence is replaced by a hole~$\ctx$. In the intuitionistic setting we have to distinguish between \defn{positive contexts}, denoted by $\cP$, and \defn{negative contexts}, denoted by $\cN$, depending on the position of the `hole':
\begin{equation}\label{eq:contexts}
  \hskip-1em
  \begin{array}{r@{~\coloneq~}r@{\,}r@{\,\mid\,}r@{\,\mid\,}r@{\,\mid\,}r@{\,\mid\,}r@{\,\mid\,}r}
    \cP&\ctx\mid& \cP\ltens A& A\ltens\cP& A\limp\cP&\cN\limp A& \cP\lseq A& A\lseq\cP
    \\
    \cN&&\cN\ltens A& A\ltens\cN& A\limp\cN& \cP\limp A& \cN\lseq A& A\lseq\cN
  \end{array}
\end{equation}

The inference rules for the \defn{system $\IBV$} are shown in \Cref{fig:IBV}. 
The reader familiar with classical $\BV$ might be surprised at the large number of inference rules. 
But note that 
(i) because we are in the intuitionistic setting, we need two versions of each rule: one for positive contexts and one for negative contexts,%
\footnote{This would correspond to rules on the left and rules on the right of the turnstile in the sequent calculus.} 
(ii) because the $\limp$ is not commutative, we need a left and a right version of each $\swir$-rule (called \defn{switch}), and 
(iii) because the $\lunit$ is not a proper unit of the $\lseq$, we need the different versions of the $\refrule$- and $\sqir$-rules, which would just be instances of the $\qdr$-rule in classical $\BV$.%
\footnote{See \cite{kahr:04:iscis,tito:str:SIS-III,retore:99} for unit-free versions of $\BV$, also having these rules.}

We use the $\circ$ and $\bullet$ decoration on the inference rules to indicate whether it applies in a positive or negative context, respectively. Finally, the rules below the dashed line have no such decoration, which indicates that they can be applied in positive and negative contexts. Furthermore, we use dotted lines for the rules to indicate that they correspond to what is usually given as equational theory in classical $\BV$. They comprise associativity, commutativity, unit-equations, and currying.\looseness=-1 

From the inference rules, we can now define derivations. We are going to use the \emph{open deduction} style, and again, because of the intuitionistic setting, we have \defn{positive} and \defn{negative derivations}, which are defined inductively as follows:
\begin{equation}\arraycolsep=3pt
  \scalebox{.9}{$
  \begin{array}{l@{\;}rc|c|c|c|c|c}
    \mbox{Positive:}& \dDp, \dDq \coloneq& \oodh{A} & \oodh{\dDp\ltens \dDq} & \oodh{\dDp\lseq \dDq} & \oodh{\dDn \limp \dDp} & \odn{\dDp}{\rrule}{\dDq}{\wrule}& \odt{\dDp}{\rrule}{\dDq}{}
  \\\\[-1ex]
    \mbox{Negative:~}& \dDn, \dDm \coloneq& \oodh{A} & \oodh{\dDn\ltens \dDm} & \oodh{\dDn\lseq \dDm} & \oodh{\dDp \limp \dDn} & \odn{\dDn}{\rrule}{\dDm}{\brule} & \odt{\dDn}{\rrule}{\dDm}{}
  \end{array}
  $}
\end{equation}
where $\oodh{A}$ is the identity derivation with premise $A$ and conclusion $A$. The composition by $\rrule$ is only allowed if the premise of $\rrule$ is the conclusion of the derivation $\dDp$ \resp{$\dDn$}, and the conclusion of $\rrule$ is the premise of the derivation $\dDq$ \resp{$\dDm$}. Then the premise of the resulting derivation is the premise of $\dDp$ \resp{$\dDn$} and its conclusion is the conclusion of $\dDq$ \resp{$\dDm$}. 

Let $\XS$ be a set of inference rules. We write ${\oddrS[\dDp]AB\XS}$ (resp.~${\oddbS[\dDn]AB\XS}$) for a positive derivation $\dDp$ (resp.~negative derivation~$\dDn$) whose inference rules are all from $\XS$, whose premise is $A$ and whose conclusion is~$B$.
A \defn{proof} of $A$ is a positive derivation $\dDp$ with premise $\lunit$ and conclusion $A$, which may be denoted $\od{\odp{\dDp}{A}{\XS}}$, and we write $\proves[\XS] A$ if there is a proof of $A$ in $\XS$.
In the following, we will omit $\XS$ in derivations if $\XS=\BiV$.

Consider now the inference rules in Figure~\ref{fig:SIBV}. They are the up-versions of the down-rules (i.e., the ones with a down-arrow in the name) of $\IBV$. The rules of $\IBV$ without an arrow in the name are part of both, the up- and the down-fragment. We call \defn{system $\SIBV$} the union of the rules in Figures~\ref{fig:IBV} and~\ref{fig:SIBV}. 
The `$\sS$' in $\SIBV$ stands for \emph{symmetric}, and we follow here the naming scheme that is common in deep inference \cite{gug:str:01,bru:tiu:01,roversi:lambda}.

Of course, in the next sections, we will show that the two systems $\SIBV$ and $\IBV$ are equivalent, i.e., prove the same formulas: $\proves[\SIBV] A$ iff $\proves[\IBV] A$.

\begin{remark}
  It is common in deep inference, and in particular for $\BV$, to consider formulas modulo an equivalence relation, covering associativity, commutativity, and units for the connectives. In our system, this would correspond to the rules marked with a dotted line in Figures~\ref{fig:IBV} and~\ref{fig:SIBV}. We chose not to do this in this paper because (i) the unit $\lunit$ is only half a unit for the $\lseq$ connective, and the rules $\wsudr$, $\wcudr$, $\bsudr$, $\bcudr$ would be needed anyway, (ii) it might be confusing to the reader to have implicit currying, and (iii) making everything explicit makes derivations easier to read.
\end{remark}

%%%%%%%%%%%%%%%%%%%%%%%%%%%%%%%%%%%%%%%%%%%%%%%%%%%%%%%%%%%%%%%%
%%%%%%%%%%%%%%%%%%%%%%%%%%%%%%%%%%%%%%%%%%%%%%%%%%%%%%%%%%%%%%%%
%%%%%%%%%%%%%%%%%%%%%%%%%%%%%%%%%%%%%%%%%%%%%%%%%%%%%%%%%%%%%%%%
\section{Properties of the Systems}\label{sec:properties}
%%%%%%%%%%%%%%%%%%%%%%%%%%%%%%%%%%%%%%%%%%%%%%%%%%%%%%%%%%%%%%%%
%%%%%%%%%%%%%%%%%%%%%%%%%%%%%%%%%%%%%%%%%%%%%%%%%%%%%%%%%%%%%%%%
%%%%%%%%%%%%%%%%%%%%%%%%%%%%%%%%%%%%%%%%%%%%%%%%%%%%%%%%%%%%%%%%

In this section we show some basic properties of $\IBV$ and $\SIBV$. For this, consider the two rules\vadjust{\vskip-1ex}
\begin{equation}\label{eq:cut}
  \nodn{\lunit}{\widr}{A \limp A}{\wrule}
  \qquad\quand\qquad
  \nodn{A\limp A}{\biur}{\lunit}{\brule}
\end{equation}
called \defn{identity} and \defn{cut}. Their atomic versions $\aidr$ and $\aiur$ have already been shown in Figures~\ref{fig:IBV} and~\ref{fig:SIBV}.

%%%%%%%%%%%%%%%%%%%%%%%%%%%%%%%%%%%%%%%%%%%%%%%%%%%%%%%%%%%%%%%%
% Fig derivability of Irule
%%%%%%%%%%%%%%%%%%%%%%%%%%%%%%%%%%%%%%%%%%%%%%%%%%%%%%%%%%%%%%%%
\begin{figure}[t]
  \adjustbox{max width=\textwidth}{$\begin{array}{c}
    \begin{array}{c}
      \odt{\lunit}{\ludr}{\lunit\limp\lunit}{}
    \qquad
      \odn{\lunit}{\waidr}{a\limp a}{}
    \\\\
    \od{
      \odi{
        \odi{\odh{\lunit}}{\wsudr}{
          \oddrP[\IH]{\lunit}{B\limp B}{}
          \lseq 
          \oddrP[\IH]{\lunit}{C\limp C}{}
        }{\wrule}
      }{\wqdr}
          {
        (B\lseq C) \limp (B\lseq C)
      }{\wrule}
    }
    \end{array}
    \quad
    \od{
      \odo{
        \odi{
          \odo{\odh{\lunit}}{\tudr}{
              \oddrP[\IH]{\lunit}{B\limp B}{}
              \ltens 
              \oddrP[\IH]{\lunit}{C\limp C}{}
          }{}
        }{\wlsr}{
          \left(\od{
            \odo{
              \odi{
                \odh{(B\limp B)\limp C} 
              }{\blsr}{
                B\ltens (B\limp C)
              }{\brule}
            }{\tcomr}{
              (B\limp C)\ltens B
            }{}
          }\right)
          \limp C
        }{\wrule}
      }{\curry}{
        (B\limp C) \limp (B \limp C)
      }{}
    }
    \quad
    \vlupsmash{
    \od{
      \odo{
        \odi{
          \odo{\odh{\lunit}}{\tudr}{
              \oddrP[\IH]{\lunit}{B\limp B}{}
              \ltens 
              \oddrP[\IH]{\lunit}{C\limp C}{}
          }{}
        }{\wtsr}{
          B\limp 
          \odtP{
            \odt{
              B\ltens (C\limp C)
            }{\tcomr}{
                (C\limp C) \ltens B
            }{}
          }{\curryrule}{
            C\limp 
            \odtP{C\ltens B}{\tcomr}{B\ltens C}{}
          }{}
        }{\wrule}
      }{\curryrule}{
        (B\ltens C) \limp (B \ltens C)
      }{}
    }
    }
  \end{array}$}
  \caption{Derivability of $\widr$ in $\BiV$.}
  \label{fig:genIrule}
\end{figure}
%%%%%%%%%%%%%%%%%%%%%%%%%%%%%%%%%%%%%%%%%%%%%%%%%%%%%%%%%%%%%%%%
%%%%%%%%%%%%%%%%%%%%%%%%%%%%%%%%%%%%%%%%%%%%%%%%%%%%%%%%%%%%%%%%
\begin{figure}[!t]
  $$
  \scalebox{.9}{$
    \od{
      \odo{\odi{\odo{\odh{P}}{\tudr}{
            P
            \ltens
            \odnP{\lunit}{\widr}{
              \odn{Q}{\bqdr}{P}{\brule}
              \limp 
              Q
            }{\wrule}
          }{}
        }{\wlsr}{
          \odnP{P\limp P}{\biur}{\lunit}{\brule}
          \limp 
          Q
        }{\wrule}
      }{\luur}{Q}{}
    }
  {\hskip4em}
    \od{
      \odo{
        \odi{\odo{\odh{R}}{\ludr}{
            \odnP{\lunit}{\widr}{
              S
              \limp 
              \odn{S}{\wqdr}{R}{\wrule}
            }{\wrule}
            \limp 
            R
          }{}
        }{\blsr}{
          S \ltens 
            \odnP{
              R
              \limp 
              R
            }{\biur}{\lunit}{\brule}
        }{\brule}
      }{\tuur}{
        S
      }{}
    }
    $}
  $$  
  \caption{Using $\biur$ to derive $\wqur$ and $\bqur$}
  \label{fig:derive-up}
\end{figure}

\begin{lemma}\label{lemma:genIrules}
  We have the following results:
  \begin{enumerate}
    \item\label{genIrules:1} $\widr$ is derivable in $\BiV$;
    \item\label{genIrules:2} $\biur$ is derivable in $\SIBV$; 
    \item\label{genIrules:3} $\baiur$, $\wqur$, and $\bqur$ are derivable in $\BiV\cup\set{\biur,\tuur,\luur}$.
  \end{enumerate}
\end{lemma}
\begin{proof}
  The first point is proven by induction on $A$ as shown in \Cref{fig:genIrule}.
  The second point is proven dually. For the third point, $\baiur$ is a special case of $\biur$, and derivability of $\wqur$ and $\bqur$ follows via the derivations shown in ~\Cref{fig:derive-up}, 
  where $P$ and $Q$ \resp{$R$ and $S$} are the premise and conclusion of $\wqur$ \resp{$\bqur$}.
  \qed
\end{proof}

\begin{lemma}\label{lemma:TFA}
  The following are equivalent:
  \begin{enumerate}
    \item $\proves[\BiV]A\limp B$;
    \item $\proves[\BiV]\cP[A]\limp \cP[B]$ for any positive context $\cP$;
    \item $\proves[\BiV]\cN[B]\limp \cN[A]$ for any negative context $\cN$;
  \end{enumerate}
\end{lemma}
\begin{proof}
  Note that 1 is a special case of 2. We prove $1\Rightarrow 2$ and $1\Rightarrow 3$ and $2\Leftrightarrow 3$ simultaneously by induction on the context, see \eqref{eq:contexts}. We show only two cases:
  $$\adjustbox{max width=.9\textwidth}{$
  \od{
    \odo{
      \odi{
        \odh{
          \odt{\lunit}{\tudr}{
            \oddrP[\IH]{\lunit}{\cN[B]\limp \cN[A]}
            \ltens
            \odnP{\lunit}{\widr}{C\limp C}{}
          }{}
        }
      }{\wlsr}{
        \left(\od{
          \odo{
            \odi{
              \odh{(\cN[B]\limp \cN[A])\limp C} 
            }{\blsr}{
              \cN[B]\ltens (\cN[A]\limp C)
            }{\brule}
          }{\tcomr}{
            (\cN[A]\limp C)\ltens \cN[B]
          }{}
        }\right)
        \limp C
      }{\wrule}
    }{\curry}{
      (\cN[A]\limp C) \limp (\cN[B] \limp C)
    }{}
  }
  \qquad\qquad
  \od{
    \odd{\odh{\lunit}}{\IH}{
      \odt{
        \left(\od{ 
          \odo{
            \odi{
              \odo{
                \odh{\cP[A]} 
              }{}{
                \odnP{\lunit}{\widr}{C\limp C}{} \limp \cP[A]
              }{}
            }{\blsr}{
                C \ltens (C\limp \cP[A])
            }{\brule}
          }{\tcomr}{
            (C\limp \cP[A]) \ltens C
          }{}
        }\right)
        \limp
        \cP[B]
      }{\curry}{
        (C\limp \cP[A]) \limp (C\limp \cP[B])
      }{\wrule}
    }{}
  }
  $}$$
  where we have to switch between 2 and 3 because in the cases $\cN = \cP \limp C$ and $\cP = \cN \limp C$ the left-hand side of the $\limp$ has opposite polarity.
  \qed
\end{proof}

\begin{theorem}[Deduction Theorem]\label{thm:deduction}
  Let $A$ and $B$ be formulas. Then
  $$\proves[\SBiV]A\limp B  \iff \mbox{ there is a derivation } \od{\odd{\odh{A}}{\dDp}{B}{\SBiV}} \mydot$$
\end{theorem}
\begin{proof}
  Via the following two derivations:
  \begin{equation}\label{eq:deduction}
    \tag*{\normalsize\qed}
    \scriptsize
    \od{
      \odo{
        \odi{
          \odo{\odh{A}}{\tudr}{
            A\ltens 
            \left(\od{\odd{\odh{\lunit}}{}{A\limp B}{\SBiV}}\right)
            }{}
        }{\wtsr}{
          \odnP{A\limp A}{\biur}{\lunit}{\brule}
          \limp B
        }{\wrule}
      }{\luur}{
        B
      }{}
    }
    \qquad\qquad
    \od{
      \odi{\odh{\lunit}}{\widr}{
        A\limp \left(\od{\odd{\odh{A}}{}{B}{\SBiV}}\right)
      }{\wrule}
    }
  \end{equation}
\end{proof}
\begin{corollary}\label{cor:deduction}
  If $\proves[\SBiV]A\limp B$, then 
  \begin{itemize}
    \item there is a derivation in $\SBiV$ with premise $\cP[A]$ and conclusion $\cP[B]$ for any positive context $\cP$;
    \item there is a derivation in $\SBiV$ with premise $\cN[B]$ and conclusion $\cN[A]$ for any negative context $\cN$;
  \end{itemize}
\end{corollary}
\begin{proof}
  Consequence of the Deduction Theorem and \Cref{lemma:TFA}.
  \qed
\end{proof}

We conclude this section by proving that the \emph{Modus Ponens} is a valid logical inference in $\SBiV$.
\begin{corollary}[Modus Ponens]\label{cor:MP}
  Let $A$ and $B$ be formulas.
  If $\proves[\SBiV]A\limp B$ and $\proves[\SBiV]A$, then $\proves[\SBiV] B$.
\end{corollary}
\begin{proof}
  By hypothesis, we have a proof $\dDp_A$ with conclusion $A$, and the Deduction Theorem ensures the existence of a derivation $\dDp_{A\limp B}$ with premise $A$ and conclusion $B$.
  We conclude by composing `vertically' the two derivations.
\qed
\end{proof}

%%%%%%%%%%%%%%%%%%%%%%%%%%%%%%%%%%%%%%%%%%%%%%%%%%%%%%%%%%%%%%%%
%%%%%%%%%%%%%%%%%%%%%%%%%%%%%%%%%%%%%%%%%%%%%%%%%%%%%%%%%%%%%%%%
%%%%%%%%%%%%%%%%%%%%%%%%%%%%%%%%%%%%%%%%%%%%%%%%%%%%%%%%%%%%%%%%
\section{Cut Elimination}\label{sec:splitting}
%%%%%%%%%%%%%%%%%%%%%%%%%%%%%%%%%%%%%%%%%%%%%%%%%%%%%%%%%%%%%%%%
%%%%%%%%%%%%%%%%%%%%%%%%%%%%%%%%%%%%%%%%%%%%%%%%%%%%%%%%%%%%%%%%
%%%%%%%%%%%%%%%%%%%%%%%%%%%%%%%%%%%%%%%%%%%%%%%%%%%%%%%%%%%%%%%%

In this section we show that the cut rule $\biur$, given in \eqref{eq:cut}, is admissible for $\IBV$.
To prove this, we will first show that the two systems $\SBiV$ and $\BiV$ are equivalent, that is, any formula provable in $\SBiV$ is also provable in $\BiV$.
In other words, we will show that all up-rules shown in Figure~\ref{fig:SIBV} are admissible for $\IBV$.
Then the admissibility of cut will follow from Lemma~\ref{lemma:genIrules}.

\begin{restatable}{theorem}{uprules}\label{thm:uprules}
  The rules 
  $\bsudr$, $\bcudr$, $\tuur$, $\luur$,
  $\baiur$, 
  $\wqur$, and $\bqur$
  are admissible for $\IBV$.
\end{restatable}

\begin{corollary}
  The systems  $\SBiV$ and $\BiV$ are equivalent.
\end{corollary}

\begin{corollary}\label{cor:iur}
  The rule  $\biur$ is admissible for $\IBV$.
\end{corollary}

It therefore remains to prove Theorem~\ref{thm:uprules}. For this, we will use a \emph{splitting lemma} as common for deep inference systems~\cite{gug:SIS,SIS-V,lutz:phd,acc:FSCD22}. But to our knowledge, this is the first version of a splitting lemma for an intuitionistic system, and there are certain subtleties, for example, the number of cases doubles, for the same reason as the number of inference rules doubles. Here is the statement of the general splitting lemma for $\IBV$:

\begin{restatable}[Splitting]{lemma}{splitting}\label{lemma:splitting}
  Let $A$, $B$ and $K$ be formulas.
  \begin{itemize}
    \item 
    If $\proves[\BiV]K \limp (A\ltens B)$, then 
    there are formulas $K_A$ and $K_B$ such that
    $$
    \vlsmash{\oddb{K_A \ltens K_B}K}
    \quand
    \odpr{K_A \limp A}
    \quand
    \odpr{K_B \limp B}
    \mydot
    $$
    \item 
    If $\proves[\BiV](A\limp B) \limp K$, then 
    there are formulas $K_A$ and $K_B$ such that 
    $$
    \vlsmash{\oddr{K_A \limp K_B}K}
    \quand
    \odpr{K_A \limp A}
    \quand
    \odpr{B\limp K_B }
    \mydot
    $$
    \item 
    If $\proves[\BiV]K \limp (A\lseq B)$,  then 
    there are formulas $K_A$ and $K_B$ such that
    $$
    \vlsmash{\oddb{K_A \lseq K_B}K}
    \quand
    \odpr{K_A \limp A}
    \quand
    \odpr{K_B \limp B}
    \mydot
    $$

    \item 
    If $\proves[\BiV](A\lseq B) \limp K$, then 
    there are formulas $K_A$ and $K_B$ such that
    $$
    \vlsmash{\oddr{K_A \lseq K_B}K}
    \quand
    \odpr{A\limp K_A}
    \quand
    \odpr{B\limp K_B}
    \mydot
    $$
  \end{itemize}
\end{restatable}
\begin{proof}
  The proof is by induction on the derivation, considering the bottom-most (non-dotted) rule instance in the derivation.
  In \Cref{fig:splitting} we show the most general forms of the non-trivial cases to take into account. In that figure, we abuse the dotted line notation to indicate that there might be more than one application of a `dotted line' inference rule from~\Cref{fig:IBV}. 
  More details are provided in \Cref{app:splitting}.
  \qed
\end{proof}

%%%%%%%%%%%%%%%%%%%%%%%%%%%%%%%%%%%%%%%%%%%%%%%%%%%%%%%%%%%%%%%%
% Fig Splitting
%%%%%%%%%%%%%%%%%%%%%%%%%%%%%%%%%%%%%%%%%%%%%%%%%%%%%%%%%%%%%%%%
\begin{figure}[t]
  \adjustbox{max width=\textwidth}{$\begin{array}{cc}
    \odt{
      K_1\limp
        \odnP{
          (K_2\limp(A_2\ltens B_2))\ltens (A_1 \ltens B_1)
        }{\wtsr}{
          K_2\limp((A_2\ltens B_2)\ltens (A_1 \ltens B_1))
        }{\wrule}
    }{}{
      (K_1\ltens K_2)\limp((A_1\ltens A_2)\ltens (B_1 \ltens B_2))
    }{}
  &
    \odt{
      K_1\limp 
        \odnP{
          (A_1\ltens B_1)\ltens((A_2\ltens B_2)\limp K_2)
        }{\wlsr}{
          ((A_1\ltens B_1)\limp(A_2\ltens B_2))\limp K_2
        }{\wrule}
    }{}{
      ((A_1\ltens A_2)\limp(B_1\ltens B_2))\limp (K_1\ltens K_2)
    }{}
  \\[20pt]
    \text{Case 1}&  \text{Case 2}
  \\\\
    \odt{
      \odnP{
        ((K_1\limp (A_1\ltens B_1))\limp (A_2\limp B_2))
      }{\blsr}{
        (K_1\ltens ((A_1\ltens B_1)\limp (A_2\limp B_2))) 
      }{\brule}
      \limp K_2
    }{}{
      (A_1\ltens A_2)\limp(B_1\limp B_2)\limp (K_1\limp K_2)
    }{}
   &
    \odt{
      \odnP{
        (A_1\ltens B_1)\limp ((A_2\limp B_2)\ltens K_2)
      }{\btsr}{
        ((A_1\ltens B_1)\limp (A_2\limp B_2))\ltens K_2
      }{\brule}
      \limp K_1
    }{}{
      ((A_1\ltens A_2)\limp(B_1\limp B_2))\limp (K_2\limp K_1)
    }{}
  \\[20pt]
    \text{Case 3}&  \text{Case 4}
  \\\\
    \odt{
      K_1
      \limp 
      \odnP{
        (K_2\limp A_1)\lseq(K_3\limp( A_3\lseq B))
      }{\wqdr}{
        (K_2\lseq K_3)\limp(A_1\lseq(A_2\lseq B))
      }{\wrule}
    }{}{
      (K_1\ltens(K_2\lseq K_3))\limp((A_1\lseq A_2)\lseq B)
    }{}
  &
    \odt{
      K_1
      \limp 
      \odnP{
        (K_2\limp (A\lseq B_1))\lseq(K_3\limp B_2)
      }{\wqdr}{
        (K_2\lseq K_3)\limp((A\lseq B_1)\lseq B_2)
      }{\wrule}
    }{}{
      (K_1\ltens(K_2\lseq K_3))\limp(A\lseq (B_1\lseq B_2))
    }{}
  \\[20pt]
    \text{Case 5 (left associativity)}
    &
    \text{Case 5 (right associativity)}
  \\\\
    \odt{
      \odnP{
        (A_1\ltens K_1)\lseq ((A_2\lseq B)\ltens K_2)
      }{\bqdr}{
        (A_1\lseq(A_2\lseq B))\limp( K_1\lseq K_2)
      }{\brule}
      \limp K_3
    }{}{
      ((A_1\lseq A_2)\lseq B)\limp ((K_1\lseq K_2)\limp K_3)
    }{}
  &
    \odt{
      \odnP{
        ((A\lseq B_1)\ltens K_1)\lseq (B_2 \ltens K_2)
      }{\bqdr}{
        (A\lseq (B_1\lseq B_2))\limp( K_1\lseq K_2)
      }{\brule}
      \limp K_3
    }{}{
      (A\lseq (B_1\lseq B_2))\limp ((K_1\lseq K_2)\limp K_3)
    }{}
  \\[20pt]
  \text{Case 6 (left associativity)}
  &
  \text{Case 6 (right associativity)}
  \\\\
    \odt{
      K_1\limp 
      \odnP{
        (A_1\limp K_2)\lseq((A_2\lseq B)\limp K_3)
      }{\wqdr}{
        ((A_1\lseq A_2)\lseq B) \limp (K_2\lseq K_3)
      }{\wrule}
    }{}{
      ((A_1\lseq A_2)\lseq B) \limp (K_1\limp  (K_2\lseq K_3))
    }{}
  &
    \odt{
      K_1\limp
      \odnP{
        ((A\lseq B_1)\limp K_2)\lseq(B_2 \limp K_3)
      }{\wqdr}{
        (A\lseq (B_1\lseq B_2)) \limp (K_2\lseq K_3)
      }{\wrule}
    }{}{
      (A\lseq (B_1\lseq B_2)) \limp (K_1\limp  (K_2\lseq K_3))
    }{}
  \\[20pt]
    \text{Case 7 (left associativity)}
    &
    \text{Case 7 (right associativity)}
  \end{array}$}
  \caption{Cases for the splitting lemma.}
  \label{fig:splitting}
\end{figure}
%%%%%%%%%%%%%%%%%%%%%%%%%%%%%%%%%%%%%%%%%%%%%%%%%%%%%%%%%%%%%%%%
%%%%%%%%%%%%%%%%%%%%%%%%%%%%%%%%%%%%%%%%%%%%%%%%%%%%%%%%%%%%%%%%

We also need a version of splitting for the atoms.

\begin{lemma}[Atomic Splitting]\label{lemma:atomicSplitting}
  Let $a$ be an atom and $K$ a formula.
  Then:
  \begin{enumerate}
    \item\label{atomicSplitting:1} if $\proves[\BiV]K\limp a$, then there is a negative derivation ${\oddb aK}$;
    \item\label{atomicSplitting:2} if $\proves[\BiV]a\limp K$, then there is a positive derivation ${\oddr aK}$.
  \end{enumerate}
\end{lemma}
\begin{proof}
  Again, by induction on the derivation and a case analysis on the
  bottom-most rule instance. We have the following non-trivial cases:
  $$
  \adjustbox{max width=\textwidth}{$
    \odn{K_1\ltens (K_2\limp a)}{\wlsr}{(K_1\limp K_2)\limp a}{\wrule}
  \quad
    \odn{(a\limp K_1)\ltens K_2}{\wtsr}{a\limp (K_1\ltens K_2)}{\wrule}
  \quad
    \odn{(a\limp K_1)\lseq K_2}{\wrsq}{a\limp (K_1\lseq K_2)}{\wrule}
  \quad
  \odn{K_1\limp (a\limp K_2)}{\wlsq}{a\limp (K_1\lseq K_2)}{\wrule}
  $}
  $$
  We conclude by applying the inductive hypothesis on the premise of the rule after remarking that, 
  because of splitting, if $\proves[\BiV] A \ltens B$ or $\proves[\BiV] A\lseq B$, then $\proves[\BiV]A$ and $\proves[\BiV]B$ must hold.
  The derivations to prove the statement are recursively defined as follows:
  \begin{equation}
    \tag*{\qed}
  \adjustbox{max width=\textwidth}{$
  \begin{array}{c|c}
    Point~\ref{atomicSplitting:1} & Point~\ref{atomicSplitting:2}
  \\\hline
    \odt{a}{\ludr}{
      \oddrP{\lunit}{K_1}
      \limp 
      \oddbP[\IH]{a}{K_2}
    }{}
  \quad&\quad
    \odt{a}{\tudr}{
      \oddrP[\IH]{a}{K_1}
      \ltens
      \oddrP{\lunit}{K_2}
    }{}
  \quad
    \odn{a}{\wsudr}{
      \oddrP{\lunit}{K_1}
      \lseq
      \oddrP[\IH]{a}{K_2}
    }{\wrule}
  \quad
    \odn{a}{\wcudr}{
      \oddrP[\IH]{a}{K_1}
      \lseq
      \oddrP{\lunit}{K_2}
    }{\wrule}
  \end{array}
  $}
  \end{equation}
\end{proof}

In order to use the splitting lemmas for proving Theorem~\ref{thm:uprules}, we need to be able to employ them in arbitrary contexts. For this, we use the context reduction lemma.

\begin{restatable}[Context Reduction]{lemma}{contextReduction}\label{lemma:contextReduction}
  Let $A$ be a formula, $\cP$ a positive context, and $\cN$ a negative context.
  \begin{enumerate}
    \item\label{contextReduction:1} 
    If $\proves[\BiV]\cP[A]$, then there is a formula $K$ and derivations
    $$
    \vlsmash{\oddr[\dDp_X]{K\limp X}{\cP[X]}} \quand \odpr[\dDp_A]{K\limp A}
    \quad \mbox{for any formula $X$.}
    $$

    \item\label{contextReduction:2} 
    If $\proves[\BiV]\cN[A]$, then there is a formula $K$ and derivations
    $$
    \vlsmash{\oddr[\dDp_X]{X\limp K}{\cN[X]}} \quand \odpr[\dDp_A]{A\limp K}
    \quad \mbox{for any formula $X$.}
    $$
  \end{enumerate}
\end{restatable}
\begin{proof}
  By induction on the structure of (positive and negative) contexts, see \eqref{eq:contexts}.
  Except for the trivial cases $\cP=H\limp \ctx$ and $\cN=\ctx\limp H$, the discussion of each case requires to apply the splitting lemma once to get derivations to be used, together with the inductive case, to construct the derivations $\dDp_X$ and $\dDn_X$.
  More details are provided in \Cref{app:splitting}.
  \qed
\end{proof}

\begin{proof}[for Theorem~\ref{thm:uprules}]
  We can remove all instances of the up-rules in a derivation, starting with a topmost one. Given such an instance $\rur$, we apply context reduction and splitting to the premise or $\rur$ (once for $\bsudr$, $\bcudr$, $\tuur$, $\luur$, and twice for $\baiur$, $\bqur$, $\wqur$). The derivations we obtained 
  are then assembled to construct a derivation of the conclusion of $\rur$. For $\baiur$ and $\wqur$, these derivations are shown in \Cref{fig:upElim}.
  More details are given in \Cref{app:splitting}.
  \qed
\end{proof}

%%%%%%%%%%%%%%%%%%%%%%%%%%%%%%%%%%%%%%%%%%%%%%%%%%%%%%%%%%%%%%%%
% Fig Up Elimination
%%%%%%%%%%%%%%%%%%%%%%%%%%%%%%%%%%%%%%%%%%%%%%%%%%%%%%%%%%%%%%%%
\begin{figure}[t!]
  \adjustbox{max width=\textwidth}{$\begin{array}{c|c}
    \od{
      \odo{
        \odi{\odh{\lunit}}{\waidr}{
          \oddr[\dD_K]{
            {\oddb[\dD_L]{a}{K_L}
            \limp 
            \oddr[\dD_R]{a}{K_R}}
          }{K}
        }{\wrule}
      }{\ludr}{
        \oddr[\dD_\lunit]{\lunit\limp K}{\cN[\lunit]}
      }{}
    }
  \quad&\quad
    \od{
      \odd{
        \odo{\odi{
          \odo{\odh{\lunit}}{\wsudr}{
            \oddrP[\dD_A]{\lunit}{K_A\limp A}\lseq \oddrP[\dD_C]{\lunit}{K_C\limp C}
          }{}
        }{\wqdr}{
          (K_A\lseq K_C) 
          \limp
          \odnP{
            \left(
              \od{\odi{\odo{\odh A}{\tudr}{\oddrP[\dD_B]{\lunit}{K_B \limp B} \ltens A}{}
                }{\wtsr}{
                  K_B \limp (B\ltens A)
                }{\wrule}
              }
            \right)
            \lseq 
            \left(
              \od{\odi{\odo{\odh C}{\tudr}{\oddrP[\dD_D]{\lunit}{K_D \limp D} \ltens C}{}
                }{\wtsr}{
                  K_D \limp (D\ltens C)
                }{\wrule}
              }
            \right)
          }{\wqdr}{
            (K_B\lseq K_D)\limp ((B\ltens A)\lseq(D\ltens C))
          }{\wrule}
        }{\wrule}}{\curryrule}{
          \oddbP[\dD_K]{
            \oddbP{K_A\lseq K_C}{K_L}
            \ltens 
            \oddbP{K_B\lseq K_D}{K_R}
          }{K}
          \limp 
          \left(
            \odtP{B\ltens A}{\tcomr}{A \ltens B}{}
            \lseq
            \odtP{D\ltens C}{\tcomr}{C \ltens D}{}
          \right)
        }{}
      }{\dDp_{(A\ltens B)\lseq(C\ltens D)}}{
        \cP[(A\ltens B)\lseq(C\ltens D)]
      }{}
    }
  \end{array}$}
  \caption{
    Derivations for admissibility of $\baiur$ and $\wqur$. The case for $\bqur$ is similar. 
    All subderivations exist by splitting and context reduction lemmas.
  }
  \label{fig:upElim}
\end{figure}
%%%%%%%%%%%%%%%%%%%%%%%%%%%%%%%%%%%%%%%%%%%%%%%%%%%%%%%%%%%%%%%%
%%%%%%%%%%%%%%%%%%%%%%%%%%%%%%%%%%%%%%%%%%%%%%%%%%%%%%%%%%%%%%%%

%%%%%%%%%%%%%%%%%%%%%%%%%%%%%%%%%%%%%%%%%%%%%%%%%%%%%%%%%%%%%%%%
%%%%%%%%%%%%%%%%%%%%%%%%%%%%%%%%%%%%%%%%%%%%%%%%%%%%%%%%%%%%%%%%
%%%%%%%%%%%%%%%%%%%%%%%%%%%%%%%%%%%%%%%%%%%%%%%%%%%%%%%%%%%%%%%%
\section{\boldmath From $\IMLL$ to $\IBV$}\label{sec:IMLLtoIBV}
%%%%%%%%%%%%%%%%%%%%%%%%%%%%%%%%%%%%%%%%%%%%%%%%%%%%%%%%%%%%%%%%
%%%%%%%%%%%%%%%%%%%%%%%%%%%%%%%%%%%%%%%%%%%%%%%%%%%%%%%%%%%%%%%%
%%%%%%%%%%%%%%%%%%%%%%%%%%%%%%%%%%%%%%%%%%%%%%%%%%%%%%%%%%%%%%%%

In this section we show that $\BiV$ is a conservative extension of intuitionistic multiplicative linear logic ($\MiLL$): 
every $\IMLL$-formula that is provable in $\IBV$ is also provable in $\IMLL$.

We assume the reader to be familiar with sequent calculi and cut-elimination.
A \defn{sequent} is a pair $\Gamma\vdash A$, where $\Gamma$ is a (possibly empty) set of occurrences of formulas, and $A$ is a formula.
\Cref{fig:IMLL} shows the inference rules of the sequent calculus for $\MiLL$, and it's well known that it admits cut elimination.
\begin{theorem}
  The rule $\AXrule$ is derivable in $\MiLL$.
\end{theorem}
\begin{theorem}\label{thm:cutElimIMLL}
  The rule $\cutr$ is admissible $\MiLL$.
\end{theorem}

%%%%%%%%%%%%%%%%%%%%%%%%%%%%%%%%%%%%%%%%%%%%%%%%%%%%%%%%%%%%%%%%
% Rules sequent calculi
%%%%%%%%%%%%%%%%%%%%%%%%%%%%%%%%%%%%%%%%%%%%%%%%%%%%%%%%%%%%%%%%
\begin{figure}[!t]
  \adjustbox{max width=\textwidth}{$
  \begin{array}{c@{\qquad}c@{\qquad}c|c}
    \vlinf{\axrule}{}{a \vdash a}{}
  &
    \vlinf{\rlimp}{}{\Gamma \vdash A\limp B}{ \Gamma, A\vdash B}
  &
    \vliinf{\llimp}{}{ \Gamma, A\limp B, \Delta \vdash C}{\Gamma
    \vdash A}{B,\Delta\vdash C}
  \quad&\quad
    \vliinf{\cutr}{}{\Gamma, \Delta \vdash  C}{ \Gamma\vdash A}{A,\Delta\vdash C}
  \\[2.5ex]
    \vlinf{\lunitr}{}{\vdash \lunit}{}
    \qquad
    \vlinf{\lunitl}{}{\Gamma,\lunit\vdash A}{\Gamma\vdash A}
  &
    \vlinf{\lltens}{}{\Gamma, A\ltens B \vdash C}{ \Gamma, A, B \vdash C}
  &
    \vliinf{\rltens}{}{ \Gamma, \Delta\vdash A\ltens B}{\Gamma \vdash A}{\Delta \vdash B}
  &
    \vlinf{\AXrule}{}{A \vdash A}{}
  \end{array}
  $}
  \caption{
    {\bf Left:} Rules for $\MiLL$
    \qquad 
    {\bf Right:} The $\cutr$ rule and the generalized axiom
  }
  \label{fig:IMLL}
\end{figure}
%%%%%%%%%%%%%%%%%%%%%%%%%%%%%%%%%%%%%%%%%%%%%%%%%%%%%%%%%%%%%%%%
%%%%%%%%%%%%%%%%%%%%%%%%%%%%%%%%%%%%%%%%%%%%%%%%%%%%%%%%%%%%%%%%

A \defn{formula interpretation} of a sequent $B_1,\ldots,B_n \vdash A$ is a formula of the shape $(B_1\ltens \cdots\ltens B_n)\limp A$ for any bracketing of the $\ltens$. If $n=0$ then the unique \defn{formula interpretation} is just $A$. 
%\begin{remark}\label{rem:formulaInt}
Note that any two formula interpretations of a sequent are interderivable in $\IBV$ using the rules $\tcomr$ and $\tassr$. This allows us to define 
a sequent $\Gamma \vdash B$ to be \defn{derivable} in $\BiV$ if any of its formula interpretations is derivable in $\BiV$.

%%%%%%%%%%%%%%%%%%%%%%%%%%%%%%%%%%%%%%%%%%%%%%%%%%%%%%%%%%%%%%%%
\begin{lemma}\label{lemma:IMLLtoIBV}
  Let $\rrule$ be an inference rule of $\MiLL$.
  If every premise of and instance of $\rrule$ is derivable in $\BiV$, then so is its conclusion.
\end{lemma}
\begin{proof}
  For the rules $\lunit$ and $\lltens$ this is trivial.
  For  $\axrule$ it follows by a single instance of $\aidr$.
  For $\rlimp$, observe that a formula interpretation of its conclusion can be obtained via  the $\curryrule$-rule from a formula interpretation of its premise.
  Finally, for the rules $\llimp$ and $\rltens$, formula interpretations of their conclusions are derived via the following derivations:
  $$
  \adjustbox{max width=\textwidth}{$
  \od{\odp{\dD'}{
    \od{
      \odo{
        \odi{\odh{
            (\Gamma\limp A) 
            \ltens 
            \odtP{
              (\Delta\ltens B) \limp C
            }{\curryrule}{
              B\limp (\Delta\limp C)
            }{}
          }
        }{\wlsr}{
          \odnP{(\Gamma\limp A)\limp (B \ltens \Delta)}{\btsr}{
            \odnP{(\Gamma\limp A)\limp B}{\blsr}{
              \Gamma\ltens (A\limp B)
            }{\brule}
            \ltens \Delta
          }{\brule}
          \limp C
        }{\wrule}
      }{\curry}{
        (\Gamma \ltens \Delta \ltens (A\limp B)) \limp C
      }{}
    }
  }{}}
  \quand
  \od{\odp{\dD'}{
    \od{\odo{\odi{
      \odh{(\Gamma \limp A)\ltens (\Delta \limp B)}
    }{\wtsr}{
      \Gamma 
      \limp 
      \left(
        \od{
          \odi{
            \odo{\odh{A\ltens (\Delta \limp B)}}{\tcomr}{
              (\Delta \limp B)\ltens A
            }{}
          }{\wtsr}{\Delta\limp (A\ltens B)}{\wrule}}
      \right)
    }{\wrule}}{\uncurry}{
      (\Gamma\ltens \Delta) \limp (A\ltens B)
    }{}}
    }{}}
  $}
  $$
  where $\dD'$ are derivations of formula interpretations of their premises.
  \qed
\end{proof}
%%%%%%%%%%%%%%%%%%%%%%%%%%%%%%%%%%%%%%%%%%%%%%%%%%%%%%%%%%%%%%%%

%%%%%%%%%%%%%%%%%%%%%%%%%%%%%%%%%%%%%%%%%%%%%%%%%%%%%%%%%%%%%%%%
\begin{lemma}\label{lemma:IBVtoIMLL}
  (i) Let $\rrule\in\BiV$, and let $A$ and $B$ be formulas not containing any occurrence of~$\lseq$. If $\nodn{A}{\rrule}{B}{\wrule}$ then $A\vdash B$ is derivable in $\MiLL$, and if $\nodn{A}{\rrule}{B}{\brule}$ then $B\vdash A$ is derivable in $\MiLL$.
  (ii) Let $\cP$ and $\cN$ be a positive and a negative $\lseq$-free context, respectively. If $A\vdash B$ is derivable in $\MiLL$, then so are $\cP[A]\vdash\cP[B]$ and $\cN[B]\vdash\cN[A]$.
\end{lemma}
\begin{proof}
  (i)
  Cases for the rule $\tudr$, $\ludr$, $\tcomr$, and $\tassr$ are the standard derivations to prove that associativity and commutativity of $\ltens$, and the fact that $\lunit$ is the unit for $\ltens$ and the left-unit of $\limp$.
  Cases for rules $\waidr$ and $\curryrule$ are the following:
  $$
  \adjustbox{max width=\textwidth}{$
  \vlderivation{
    \vlin{\llimp}{}{\lunit \vdash a\limp a}{
      \vlin{\lunitl}{}{\lunit,a\vdash a}{
        \vlin{\axrule}{}{a\vdash a}{\vlhy{}}
      }
    }
  }
  \qquad
  \vlderivation{
    \vliq{\rlimp+\lltens}{}{
      A\limp (B\limp C)\vdash (A\ltens B)\limp C
    }{
      \vliiiq{2\times\llimp}{}{A\limp (B\limp C), A,B\vdash C}{
        \vlin{\axrule}{}{A\vdash A}{\vlhy{}}
      }{
        \vlin{\axrule}{}{B\vdash B}{\vlhy{}}
      }{
        \vlin{\axrule}{}{C\vdash C}{\vlhy{}}
      }
    }
  }
  \qquad
  \vlderivation{
    \vliq{2\times \rlimp}{}{
      (A\ltens B)\limp C\vdash A\limp (B\limp C)
    }{
      \vliiiq{\llimp+\rltens}{}{(A\ltens B)\limp C, A,B\vdash C}{
        \vlin{\axrule}{}{A\vdash A}{\vlhy{}}
      }{
        \vlin{\axrule}{}{B\vdash B}{\vlhy{}}
      }{
        \vlin{\axrule}{}{C\vdash C}{\vlhy{}}
      }
    }
  }
  $}
  $$
  If $\rrule\in\set{\wlsr,\wtsr,\btsr,\blsr}$, then we have either the following derivations
  \begin{equation}\label{eq:srulesIMLL}
  \small
  \vlderivation{
    \vliq{\rlimp+\lltens}{}{
      (A\limp B)\ltens C \vdash A\limp(B\ltens C)
    }{
      \vliiiq{\llimp+\rltens}{}{
        (A\limp B),C, A\vdash B\ltens C
      }{
        \vlin{\axrule}{}{A\vdash A}{\vlhy{}}
      }{
        \vlin{\axrule}{}{B\vdash B}{\vlhy{}}
      }{
        \vlin{\axrule}{}{C\vdash C}{\vlhy{}}
      }
    }
  }
  \qquad\quad
  \vlderivation{
    \vliq{\rlimp+\lltens}{}{
      A \ltens (B\limp C) \vdash (A\limp B) \limp C
    }{
      \vliiiq{2\times\llimp}{}{
        A , B\limp C , A\limp B \vdash C
      }{
        \vlin{\axrule}{}{A\vdash A}{\vlhy{}}
      }{
        \vlin{\axrule}{}{B\vdash B}{\vlhy{}}
      }{
        \vlin{\axrule}{}{C\vdash C}{\vlhy{}}
      }
    }
  }
  \end{equation}
  (ii) is shown by induction on the context, where both statements are proved simultaneously.
  \qed
\end{proof}
%%%%%%%%%%%%%%%%%%%%%%%%%%%%%%%%%%%%%%%%%%%%%%%%%%%%%%%%%%%%%%%%

%%%%%%%%%%%%%%%%%%%%%%%%%%%%%%%%%%%%%%%%%%%%%%%%%%%%%%%%%%%%%%%%
\begin{theorem}\label{thm:BiVextensMiLL}
  $\BiV$ is a conservative extension of $\MiLL$.
\end{theorem}
\begin{proof}
  Given a proof $\pi$ in $\MiLL$, we can construct by induction a proof in $\BiV$ of the formula interpretation of its conclusion using \Cref{lemma:IMLLtoIBV}.

  Conversely, assume we have proof $\dDp_A$ in $\BiV$ of a formula $A$ in $\BiV$.
  We construct a proof $\pi_A$ in $\MiLL$ with the same conclusion by induction on the number of rules in $\dDp_A$.
  We consider as base case $\aidr$, whose conclusion derivable in $\MiLL$ using $\axrule$ and~$\rlimp$.
  For the inductive case, consider a bottommost rule instance $\rrule$ in $\dDp_A$, as shown on the left below. The desired cut-free derivation in $\MiLL$ is obtained by applying cut-elimination (\Cref{thm:cutElimIMLL}) to the derivation in the middle below:
  \begin{equation}\label{eq:de-deeping}
    \dDp_A=
    \od{
      \odp{}{\cC[\odn{B'}{\rrule}{B}{}]}{\BiV}
    }
  \;\rightsquigarrow\;
    \vlderivation{
      \vliin{\cutr}{}{\vdash \cC[B]}{
        \vlhtrl{\pi_1}{}{}{\cC[B']}{~~} 
      }{
        \vlhtrl{\pi_2}{}{}{\cC[B']\vdash \cC[B]}{~~} 
      }
    }
  \;\xRightarrow{\text{cut-elim.}}\;
    \vlderivation{\vlhtrl{\pi_A}{}{}{\vdash \cC[B]}{~~}}
  \end{equation}
  where $\cC$ is positive if $\rrule$ is a $\wsymb$-rule and negative if $\rrule$ is a $\bsymb$-rule.
  Then $\pi_1$ exists by induction hypothesis and the existence of $\pi_2$ is guaranteed by \Cref{lemma:IBVtoIMLL}.
  \qed
\end{proof}
%%%%%%%%%%%%%%%%%%%%%%%%%%%%%%%%%%%%%%%%%%%%%%%%%%%%%%%%%%%%%%%%

%%%%%%%%%%%%%%%%%%%%%%%%%%%%%%%%%%%%%%%%%%%%%%%%%%%%%%%%%%%%%%%%
%%%%%%%%%%%%%%%%%%%%%%%%%%%%%%%%%%%%%%%%%%%%%%%%%%%%%%%%%%%%%%%%
%%%%%%%%%%%%%%%%%%%%%%%%%%%%%%%%%%%%%%%%%%%%%%%%%%%%%%%%%%%%%%%%
\section{\boldmath From $\IBV$ to $\BV$}\label{sec:IBVtoBV}
%%%%%%%%%%%%%%%%%%%%%%%%%%%%%%%%%%%%%%%%%%%%%%%%%%%%%%%%%%%%%%%%
%%%%%%%%%%%%%%%%%%%%%%%%%%%%%%%%%%%%%%%%%%%%%%%%%%%%%%%%%%%%%%%%
%%%%%%%%%%%%%%%%%%%%%%%%%%%%%%%%%%%%%%%%%%%%%%%%%%%%%%%%%%%%%%%%

%%%%%%%%%%%%%%%%%%%%%%%%%%%%%%%%%%%%%%%%%%%%%%%%%%%%%%%%%%%%%%%%
% Rules IBV-
%%%%%%%%%%%%%%%%%%%%%%%%%%%%%%%%%%%%%%%%%%%%%%%%%%%%%%%%%%%%%%%%
\begin{figure}[!t]
  $$
  \vlinf{\waidr}{\wrule}{a\limp a}{}
  \quad
  \vlinf{\waidr}{\wrule}{(a\limp a)\ltens B}{B}
  % \quad
  % \vlinf{\aidr}{\brule}{(a\limp a)\ltens B}{B}
  \quad
  \vlinf{\baidr}{\brule}{(a\limp a)\limp B}{B}
  \quad
  \vlinf{\waidr}{\wrule}{(a\limp a)\lseq B}{B}
  \quad
  \vlinf{\waidr}{\wrule}{B\lseq(a\limp a)}{B}
  $$
  \caption{$\aidr$-rules for $\IBVm$.}
  \label{fig:IBVm}
\end{figure}
%%%%%%%%%%%%%%%%%%%%%%%%%%%%%%%%%%%%%%%%%%%%%%%%%%%%%%%%%%%%%%%%
%%%%%%%%%%%%%%%%%%%%%%%%%%%%%%%%%%%%%%%%%%%%%%%%%%%%%%%%%%%%%%%%

One might expect that also $\BV$ is a conservative extension of $\IBV$, in the same sense as $\MLL$ is a conservative extension of $\MiLL$. 
Unfortunately, this is not true because of the collapse of the units in $\BV$.
However, if we look at the unit-free versions of the two systems, we can get our expected result. 
The unit-free version of $\IBV$, denoted by \defn{$\IBVm$}, is obtained from $\IBV$ (see \Cref{fig:IBV}) by removing all versions of $\udr$ and by replacing the $\aidr$ by the rules shown in \Cref{fig:IBVm}.

Similarly, the \defn{$\BVm$-formulas} are the $\BV$ formulas without the unit, i.e., 
they are generated from the countable set $\atomset$ of atoms and the binary connectives par ($\lpar$), tensor ($\ltens$), and seq ($\lseq$), using the following grammar:
\begin{equation}\label{eq:BVformulas} 
  A,B \coloneq a \mid \cneg a \mid   A\ltens B\mid A\lpar B\mid A\lseq B
  \qquad\qquad
  a \in \atomset
\end{equation}
We can define the negation $\cneg{(\cdot)}$ for all $\BVm$-formulas via DeMorgan duality:
$$
\cnegneg a = a \quad
\cneg{(A\ltens B)}=\cneg A\lpar \cneg B
\quad \cneg{(A\lpar B)}=\cneg A\ltens \cneg B
\quad \cneg{(A\lseq B)}=\cneg A\lseq \cneg B
$$
The inference rules for \defn{system $\BVm$} are shown in \Cref{fig:BVm}. It is the unit-free version of $\BV$ described in~\cite{kahr:04:iscis,tito:lutz:csl22,tito:str:SIS-III}, with the mix-rule removed. The rule $\feq$ is defined by the equations capturing associativity of the connectives $\ltens$, $\lseq$, and $\lpar$, and 
commutativity of $\ltens$ and $\lpar$.

%%%%%%%%%%%%%%%%%%%%%%%%%%%%%%%%%%%%%%%%%%%%%%%%%%%%%%%%%%%%%%%%
% Rules BV
%%%%%%%%%%%%%%%%%%%%%%%%%%%%%%%%%%%%%%%%%%%%%%%%%%%%%%%%%%%%%%%%
\begin{figure}[t]
%  $$
  \adjustbox{max width=\textwidth}{$
  \begin{array}{c}
      \nodt{A}{\feq}{B}{}
    \qquad
      \nodn{}{\airule_\downarrow^{\emptyset}}{a\lpar\cneg a}{}
    \qquad
      \nodn{B}{\airule_\downarrow^{\lseq}}{(a\lpar\cneg a)\lseq B}{}
    \qquad 
      \nodn{B}{\airule_\downarrow^{\lcoseq}}{B\lseq(a\lpar\cneg a)}{}
    \qquad
      \nodn{B}{\airule_\downarrow^{\ltens}}{(a\lpar\cneg a)\ltens B}{}
    \qquad
      \nodn{A\ltens B}{\mixr}{A\lpar B}{}
    \\ \\[-.5ex]
      \nodn{A\ltens(B\lpar C)}{\swir}{(A\ltens B) \lpar C}{}
    \qquad
      \nodn{(A\lpar B)\lseq C}{\squr}{A\lpar (B\lseq C)}{}
    \qquad
      \nodn{A\lseq (B\lpar C)}{\squl}{(A\lseq C)\lpar B}{}
    \qquad
      \nodn{(A\lpar B)\lseq (C\lpar D)}{\qdr}{(A\lseq C)\lpar(B\lseq D)}{}
    \qquad
      \nodn{A\lseq B}{\refrule}{A\lpar B}{}
  \\\\\hdashline\\
    \begin{array}{c@{\qquad}c@{\qquad}c}
      A\ltens (B\ltens C) \feq (A\ltens B)\ltens C
      &
      A\ltens B \feq B\ltens A
      \\[-1ex]
      &&
      A\lseq (B\lseq C) \feq (A\lseq B)\lseq C 
      \\[-1ex]
      A\lpar (B\lpar C) \feq (A\lpar B)\lpar C
      &
      A\lpar B \feq B\lpar A 
    \end{array}
  \end{array}
  $}
  \caption{
    Rules for the system $\BVm$. The side condition for the $\feq$-rule is that $A$ and $B$ are equivalent modulo the equivalence relation generated by the relation $\feq$, defined below the dashed line.
  }
  \label{fig:BVm}
\end{figure}
%%%%%%%%%%%%%%%%%%%%%%%%%%%%%%%%%%%%%%%%%%%%%%%%%%%%%%%%%%%%%%%%
%%%%%%%%%%%%%%%%%%%%%%%%%%%%%%%%%%%%%%%%%%%%%%%%%%%%%%%%%%%%%%%%

We define an embedding $\emb{(\cdot)}$ from unit-free formulas into $\BVm$-formulas as follows:
\begin{equation}\label{eq:IBVemb}
  \emb{a}=a
  \quad
  \emb{(A \ltens B)}=\emb A\ltens\emb B
  \quad
  \emb{(A \limp B)}={(\emb A)\lneg}\lpar\emb B
  \quad
  \emb{(A \lseq B)}=\emb A\lseq\emb B
\end{equation}
This allows us to state the main theorem of this section.

\begin{restatable}{theorem}{thmBVm}\label{thm:BVm}
  Let $A$ be a unit-free formula. 
  We have~ $\proves[\IBVm]A$ ~iff~ $\proves[\BVm]{\emb A}$~.
\end{restatable}

For proving it, we define two subsets of $\BVm$-formulas, called \defn{positive formulas} and \defn{negative formulas}: 
\begin{equation}
  \label{eq:polarity}
  \hskip-1em
  \begin{array}{l@{\quad}c@{\;\coloneq\;}c@{\,\mid\,}c@{\,\mid\,}c@{\,\mid\,}c}
    \mbox{positive:}& \wA, \wB  & a & \wA\ltens \wB & \bA \lpar \wB & \wA\lseq \wB 
  \\[1ex]
    \mbox{negative:}& \bA, \bB & \cneg a & \isw A \ltens \isb B & \bA\lpar \bB & \bA\lseq \bB
  \end{array}
\end{equation}
with $a\in\atomset$.
Clearly, a $\BVm$-formula can either be positive (e.g., $a\ltens a$) or negative (e.g., $\cneg a\lpar \cneg a$) or neither (e.g., $\cneg a\ltens\cneg a$), but never both. We call a $\BVm$-formula that is neither positive nor negative \defn{unpolarisable}.

\begin{lemma}\label{lem:emb}
  A $\BVm$-formula $A$ is positive iff it is in the image of $\emb{(\cdot)}$.
\end{lemma}

\begin{proof}
  By induction on $A$, proving at the same time that $A$ is negative iff $\cneg A$ is in the image of $\emb{(\cdot)}$.
  \qed
\end{proof}

\begin{proof}[for Theorem~\ref{thm:BVm}]
  Let us first assume we have a derivation of $A$ in $\IBVm$. We can apply the mapping $\emb{(\cdot)}$ to every formula in that derivation and obtain a valid $\BVm$ derivation. For the converse, it suffices to observe that all rules of $\BVm$ either preserve the positive polarity of the whole formula (every proof starts with a positive atomic interaction), or break it. And no inference rule of $\BVm$ can transform an unpolarisable formula back into one with positive polarity. Furthermore, enumerating all possible polarity-preserving assignments to the rules in \Cref{fig:BVm} yields the rules of $\IBVm$.\footnote{This is, in fact, the method by which we found the inference rules for $\IBV$.}
  Now the result follows from Lemma~\ref{lem:emb}. (See \Cref{app:IBVtoBV} for more details)
  \qed
\end{proof}

%%%%%%%%%%%%%%%%%%%%%%%%%%%%%%%%%%%%%%%%%%%%%%%%%%%%%%%%%%%%%%%%
%%%%%%%%%%%%%%%%%%%%%%%%%%%%%%%%%%%%%%%%%%%%%%%%%%%%%%%%%%%%%%%%
%%%%%%%%%%%%%%%%%%%%%%%%%%%%%%%%%%%%%%%%%%%%%%%%%%%%%%%%%%%%%%%%
\section{Associative or not associative, that is the question}
\label{sec:assoc}
%%%%%%%%%%%%%%%%%%%%%%%%%%%%%%%%%%%%%%%%%%%%%%%%%%%%%%%%%%%%%%%%
%%%%%%%%%%%%%%%%%%%%%%%%%%%%%%%%%%%%%%%%%%%%%%%%%%%%%%%%%%%%%%%%
%%%%%%%%%%%%%%%%%%%%%%%%%%%%%%%%%%%%%%%%%%%%%%%%%%%%%%%%%%%%%%%%
\def\Nform{$\NiML$-formula\xspace}
\def\Nforms{$\NiML$-formulas\xspace}
\def\Nseq{$\NiML$-sequent\xspace}
\def\Nseqs{$\NiML$-sequents\xspace}
\def\lprec{\lseq}
\def\preseq#1{[\![#1]\!]^\lprec_{\lseq}}

%%%%%%%%%%%%%%%%%%%%%%%%%%%%%%%%%%%%%%%%%%%%%%%%%%%%%%%%%%%%%%%%
% Rules NML
%%%%%%%%%%%%%%%%%%%%%%%%%%%%%%%%%%%%%%%%%%%%%%%%%%%%%%%%%%%%%%%%
\begin{figure}[!t]
  $$
  \begin{array}{c}
      \vlinf{\axrule}{}{\vdash a,\cneg a}{}
  \quad
      \vlinf{\lpar}{}{\vdash \Gamma, A\lpar B}{\vdash \Gamma, A, B}
  \quad
      \vliinf{\ltens}{}{\vdash  \Gamma, \Delta, A\ltens B}{\vdash \Gamma,A}{\vdash \Delta, B}
  \quad
      \vliinf{\lprec}{n\geq 0}{
        \vdash \Gamma, \Delta, A_1\lprec B_1, \ldots, A_n\lprec B_n
      }{
        \vdash \Gamma , A_1, \ldots,A_n
      }{
        \vdash \Delta, B_1, \ldots,B_n
      }
  \end{array}
  $$
  \caption{Sequent calculus for unit-free $\NMLm$.}
  \label{fig:NML}
\end{figure}
%%%%%%%%%%%%%%%%%%%%%%%%%%%%%%%%%%%%%%%%%%%%%%%%%%%%%%%%%%%%%%%%
%%%%%%%%%%%%%%%%%%%%%%%%%%%%%%%%%%%%%%%%%%%%%%%%%%%%%%%%%%%%%%%%

The argument of Tiu~\cite{tiu:SIS-II} that there is no shallow proof system for $\BV$ makes crucial use of the associativity of~$\lseq$. 
This inspired~\cite{acc:man:PN} to design the logic $\NML$, which is a conservative extension of $\MLL$ with a non-commutative and \emph{non-associative} self-dual connective, and which has a cut-free sequent calculus.

In this section, we investigate the intuitionistic version of that logic, that we call \defn{$\INML$}, and which is obtained from the sequent calculus of $\IMLL$ (shown in Figure~\ref{fig:IMLL}) extended by the rule~$\lseq$ shown below: 
\begin{equation}\label{eq:precRule}
  \vliinf{\lprec}{n\geq 0}{\Gamma,\Delta, A_1\lprec B_1, \ldots, A_n\lprec B_n \vdash A\lprec B}{
    \Gamma,A_1, \ldots, A_n\vdash A
  }{
    \Delta, B_1, \ldots, B_n \vdash B
  }
\end{equation}
\begin{theorem}
  The rule $\AXrule$ is derivable in $\NiML$.
\end{theorem}
\begin{theorem}
  The rule $\cutr$ is admissible in $\NiML$.
\end{theorem}
\begin{proof}
  The proof is the same as for $\IMLL$, with one cut-elimination step added for the $\lseq$-rule:
  $$\adjustbox{max width=\textwidth}{$
    \vlderivation{
      \vliin{\cutr}{}{
        \Gamma,\Delta \vdash C\lprec D
      }{
        \vliin{\lprec}{n}{\Gamma\vdash A\lprec B}{
          \vlhy{\Gamma_1\vdash A}
        }{
          \vlhy{\Gamma_2\vdash B}
        }
      }{
        \vliin{\lprec}{m+1}{A\lprec B, \Delta \vdash C\lprec D}{
          \vlhy{A,\Delta_1 \vdash C}
        }{
          \vlhy{B,\Delta_2 \vdash D}
        }
      }
    }
  \;\rightsquigarrow\;
    \vlderivation{
      \vliin{\lprec}{n+m}{
        \Gamma,\Delta \vdash C\lprec D
      }{
        \vliin{\cutr}{}{\Gamma_1,\Delta_1\vdash C}{
          \vlhy{\Gamma_1\vdash A}
        }{
          \vlhy{A,\Delta_1 \vdash C}
        }
      }{
        \vliin{\cutr}{}{\Gamma_2, \Delta_2 \vdash D}{
          \vlhy{\Gamma_2\vdash B}
        }{
          \vlhy{B,\Delta_2 \vdash D}
        }
      }
    }
    $}$$
  where $\Gamma$ (resp.~$\Delta$) can contain formulas of the shape $E\lseq F$, which are decomposed in $\Gamma_1$ and $\Gamma_2$ (resp.~$\Delta_1$ and $\Delta_2$).
  \qed
\end{proof}
The logic $\INML$ is indeed the same logic as we would obtain by removing associativity of $\lseq$ from $\IBV$.

\begin{restatable}{theorem}{BiVextensNiML}\label{thm:BiVextensNiML}
  Let $A$ be a formula.
  Then ~$\proves[\NiML]A$ ~iff~ $\proves[\IBV\setminus\set{\sassr,\sassl}]A$~.
\end{restatable}
\begin{proof}
  Most of the work has already been done in Lemmas~\ref{lemma:IMLLtoIBV} and~\ref{lemma:IBVtoIMLL}. It remains to show that (i) Lemma~\ref{lemma:IMLLtoIBV} also applies to the rule $\lseq$ in~\eqref{eq:precRule}, and that $\sassr$, $\sassl$ are not needed for this, and (ii) that Lemma~\ref{lemma:IBVtoIMLL} also applies to the $\qdr$-, $\sqir$-, and $\refrule$-rules, when we use $\INML$ instead of $\IMLL$. 
  Details in \Cref{app:assoc}.
  \qed
\end{proof}

Unsurprisingly, we have for $\INML$ the same conservativity results as we proved for $\IBV$ in the previous two sections.

\begin{theorem}
  $\NiML$ is a conservative extension of $\MiLL$.
\end{theorem}
\begin{proof}
  It suffices to remark that all rules in $\NiML$ have the subformula property, so any proof of a $\lseq$-free formula in $\NiML$ can only use rules in $\MiLL$.
  \qed
\end{proof}

We now would like to show that $\NML$ is conservative over $\INML$, but we encounter the same problems with the units as in $\BV$. We therefore consider only unit-free formulas, and the unit-free version of $\NML$, that we call \defn{$\NMLm$}, and which is shown in~\Cref{fig:NML}.
\begin{restatable}{theorem}{NMLextensNiML}\label{thm:NMLextensNiML}
  Let $A$ be a unit-free formula.
  We have~ $\proves[\INML]A$ ~iff~ $\proves[\NMLm]{\emb{A}}$~.
\end{restatable}
\begin{proof}
  This follows in the same way as (unit-free) $\MLLx$ is conservative over unit-free $\IMLL$.
  See \Cref{app:assoc} for details.
\qed
\end{proof}

We conclude this section by observing that the cut restricted to associativity of $\lseq$ is enough to obtain from $\INML$ a sequent calculus for $\IBV$.
To this end, we introduce the following \defn{associative-cut} rules:
\begin{equation}\label{eq:asso-cut}
  \adjustbox{max width=.9\textwidth}{$
  \vliinf{\acutl}{}{
    \Gamma \vdash D
  }{
    \Gamma \vdash (A\lprec B) \lprec C 
  }{
    A\lprec (B \lprec C) , \Delta \vdash D
  }
  \qquad
  \vliinf{\acutr}{}{
    \Gamma \vdash D
  }{
    \Gamma \vdash A\lprec (B \lprec C) 
  }{
     (A\lprec B) \lprec C,\Delta \vdash D
  }
  $}
\end{equation}

\begin{restatable}{theorem}{BiVextensNiMLcut}\label{thm:BiVextensNiMLcut}
  Let $A$ be a formula.
  Then $\proves[\BiV]A$ ~iff~ $\proves[\NiML\cup\set{\acutl,\acutr}]A$.
\end{restatable}
\begin{proof}
  As before, it suffices to extend Lemmas~\ref{lemma:IMLLtoIBV} and~\ref{lemma:IBVtoIMLL}.
  Note that for simulating $\acutl$ and $\acutr$ in $\IBV$, we use $\biur$ and then apply Corollary~\ref{cor:iur}.
  \qed
\end{proof}

\begin{remark}
  A sequent calculus for $\BV$ could be obtained by adding to the cut-free sequent system $\NML$ with the unit 
  from \cite{acc:man:PN} the one-side version of the rule $\acutr$.
\end{remark}

%%%%%%%%%%%%%%%%%%%%%%%%%%%%%%%%%%%%%%%%%%%%%%%%%%%%%%%%%%%%%%%%
%%%%%%%%%%%%%%%%%%%%%%%%%%%%%%%%%%%%%%%%%%%%%%%%%%%%%%%%%%%%%%%%
%%%%%%%%%%%%%%%%%%%%%%%%%%%%%%%%%%%%%%%%%%%%%%%%%%%%%%%%%%%%%%%%
\section{Conclusion and Future Work}\label{sec:conc}
%%%%%%%%%%%%%%%%%%%%%%%%%%%%%%%%%%%%%%%%%%%%%%%%%%%%%%%%%%%%%%%%
%%%%%%%%%%%%%%%%%%%%%%%%%%%%%%%%%%%%%%%%%%%%%%%%%%%%%%%%%%%%%%%%
%%%%%%%%%%%%%%%%%%%%%%%%%%%%%%%%%%%%%%%%%%%%%%%%%%%%%%%%%%%%%%%%

In this paper, we have defined two logics, $\BiV$ and $\NiML$, which are the intuitionistic counterparts of the non-commutative logics $\BV$ and $\NML$, respectively.
In order to show that both systems are conservative extensions of the multiplicative intuitionistic linear logic $\MiLL$,
for $\BiV$ we have implemented the first splitting proof for an intuitionistic system, and for $\NiML$ we have shown that cut-elimination holds for the sequent system $\NiML$.

In future work, we hope to be able to extend our systems to obtain conservative extensions of intuitionistic logic with a non-commutative connective:
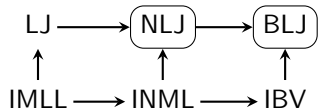
\begin{equation*}\label{eq:new}
\begin{array}{c@{\qquad}c@{\qquad}c}
  \vpz2{\LJ}&  \vmod4{\,\NLJ\,} & \vmod{6}{\,\BLJ\,}
  \\\\
  \vpz1{\MiLL} & \vpz3{\NiML} & \vpz5{\BiV}
\end{array}
\Dedges{pz1/pz2,pz1/pz3,pz3/pz5,pz3/mod4,pz2/mod4}
\Dedges{pz5/mod6,mod4/mod6}
\end{equation*}
These systems could be used as a framework for defining typing disciplines for imperative programming languages in which the sequentiality of the order of instructions is modeled with the new non-commutative connective.

%%%%%%%%%%%%%%%%%%%%%%%%%%%%%%%%%%%%%%%%%%%%%%%%%%%%%%%%%%%%%%%%
%%%%%%%%%%%%%%%%%%%%%%%%%%%%%%%%%%%%%%%%%%%%%%%%%%%%%%%%%%%%%%%%
%%%%%%%%%%%%%%%%%%%%%%%%%%%%%%%%%%%%%%%%%%%%%%%%%%%%%%%%%%%%%%%%
\bibliographystyle{plain}%alphaurl}
\bibliography{biblio}
%%%%%%%%%%%%%%%%%%%%%%%%%%%%%%%%%%%%%%%%%%%%%%%%%%%%%%%%%%%%%%%%
%%%%%%%%%%%%%%%%%%%%%%%%%%%%%%%%%%%%%%%%%%%%%%%%%%%%%%%%%%%%%%%%
%%%%%%%%%%%%%%%%%%%%%%%%%%%%%%%%%%%%%%%%%%%%%%%%%%%%%%%%%%%%%%%%

\clearpage

%%%%%%%%%%%%%%%%%%%%%%%%%%%%%%%%%%%%%%%%%%%%%%%%%%%%%%%%%%%%%%%%
%%%%%%%%%%%%%%%%%%%%%%%%%%%%%%%%%%%%%%%%%%%%%%%%%%%%%%%%%%%%%%%%
%%%%%%%%%%%%%%%%%%%%%%%%%%%%%%%%%%%%%%%%%%%%%%%%%%%%%%%%%%%%%%%%
%%%%%%%%%%%%%%%%%%%%%%%%%%%%%%%%%%%%%%%%%%%%%%%%%%%%%%%%%%%%%%%%
\appendix
%%%%%%%%%%%%%%%%%%%%%%%%%%%%%%%%%%%%%%%%%%%%%%%%%%%%%%%%%%%%%%%%
%%%%%%%%%%%%%%%%%%%%%%%%%%%%%%%%%%%%%%%%%%%%%%%%%%%%%%%%%%%%%%%%
%%%%%%%%%%%%%%%%%%%%%%%%%%%%%%%%%%%%%%%%%%%%%%%%%%%%%%%%%%%%%%%%
%%%%%%%%%%%%%%%%%%%%%%%%%%%%%%%%%%%%%%%%%%%%%%%%%%%%%%%%%%%%%%%%

%%%%%%%%%%%%%%%%%%%%%%%%%%%%%%%%%%%%%%%%%%%%%%%%%%%%%%%%%%%%%%%%
%%%%%%%%%%%%%%%%%%%%%%%%%%%%%%%%%%%%%%%%%%%%%%%%%%%%%%%%%%%%%%%%
%%%%%%%%%%%%%%%%%%%%%%%%%%%%%%%%%%%%%%%%%%%%%%%%%%%%%%%%%%%%%%%%
\section{Proof for Footnote~\ref{fn:collapse}}\label{app:collapse}
%%%%%%%%%%%%%%%%%%%%%%%%%%%%%%%%%%%%%%%%%%%%%%%%%%%%%%%%%%%%%%%%
%%%%%%%%%%%%%%%%%%%%%%%%%%%%%%%%%%%%%%%%%%%%%%%%%%%%%%%%%%%%%%%%
%%%%%%%%%%%%%%%%%%%%%%%%%%%%%%%%%%%%%%%%%%%%%%%%%%%%%%%%%%%%%%%%

\begin{proposition}\label{prop:collapse}
  If $(\lunit\lseq A)\limp A$ and  $(A\lseq \lunit)\limp A$ were provable in $\IBV$, then $\lseq$ and $\ltens$ would collapse (in the sense that $A\lseq B$ and $A\ltens B$ would be logically equivalent formulas).
\end{proposition}
\begin{proof}
  The implication $(A \ltens B) \limp (A \lseq B)$ already holds in $\BV$ and $\IBV$ (it suffices to apply a $\brr$ and a $\widr$).
  For the other implication $(A\lseq B) \limp (A\ltens B)$, assume we have  $(\lunit\lseq A)\limp A$ and $(A\lseq \lunit)\limp A$ for every $A$. By \Cref{lemma:TFA}, there should be negative derivations $\dD_{\lunit\lseq}$ and $\dD_{\lseq\lunit}$ in $\SIBV$ with premise $A$ and conclusion $\lunit\lseq A$ and $A\lseq \lunit$ respectively, for any formula $A$.
  Then, we could construct a proof of $(A\lseq B) \limp (A\ltens B)$ in $\IBV$ by applying up-rules elimination (\Cref{thm:uprules}) to the following derivation in $\SIBV$:
  $$\small
  \odn{\lunit}{\idr}{
    \odnP{
      \oddbP[\dD_{\lseq\lunit}]{A}{A\lseq \lunit}
      \ltens 
      \oddbP[\dD_{\lunit\lseq}]{B}{\lunit\lseq B}
    }{\bqdr}{
      \odtP{A\ltens \lunit}{\tuur}{A}{}
      \lseq
      \odtP{\lunit\ltens B}{\tuur}{B}{}
    }{\brule}
    \limp
    (A\ltens B)
  }{\wrule}
  $$
\qed
\end{proof}

%%%%%%%%%%%%%%%%%%%%%%%%%%%%%%%%%%%%%%%%%%%%%%%%%%%%%%%%%%%%%%%%
\section{Proofs for Section~\ref{sec:splitting}}
\label{app:splitting}
%%%%%%%%%%%%%%%%%%%%%%%%%%%%%%%%%%%%%%%%%%%%%%%%%%%%%%%%%%%%%%%%

%%%%%%%%%%%%%%%%%%%%%%%%%%%%%%%%%%%%%%%%%%%%%%%%%%%%%%%%%%%%%%%%
\splitting*
%%%%%%%%%%%%%%%%%%%%%%%%%%%%%%%%%%%%%%%%%%%%%%%%%%%%%%%%%%%%%%%%
\begin{proof}
  We proceed by induction on the given derivation. If a bottommost rule in that derivation is active only inside $K$ or $A$ or $B$, then we can apply immediately the induction hypothesis. Therefore, we need only to consider the non-trivial cases as shown in 
\Cref{fig:splitting}: , and we show how to apply the inductive hypothesis on the premises of the rules.
  \begin{enumerate}
    \item
    The bottom of the derivation is of the following form
    $$\small
      \odt{
        K_1\limp
          \odnP{(K_2\limp(A_2\ltens B_2))\ltens (A_1 \ltens B_1)}{\wtsr}{K_2\limp((A_2\ltens B_2)\ltens (A_1 \ltens B_1))}{\wrule}
      }{}{
        (K_1\ltens K_2)\limp((A_1\ltens A_2)\ltens (B_1 \ltens B_2))
      }{}
    $$
    and we can use the inductive hypothesis to get formulas $K_L$ and $K_R$ such that there are derivations of the following form
    $$\small
      \oddb[\dD_1]{K_L\ltens K_R }{K_1}
    \quand
      \odt{\odpr{K_L \limp (K_2\limp (A_2 \ltens B_2))}}{}{(K_L\ltens K_2) \limp (A_2\ltens B_2)}{}
    \quand
      \odpr{K_R\limp(A_1\ltens B_1)}
    \mydot
    $$
    By applying inductive hypothesis on $K_R\limp(A_1\ltens B_1)$ and on $(K_L\ltens K_2) \limp (A_2\ltens B_2)$, we find formulas $K_{A_1}$, $K_{A_2}$, $K_{B_1}$ and $K_{B_2}$ such that:
    $$\small
      \oddb[\dD_R]{K_{A_1}\ltens K_{B_1}}{K_R}
      \quand
      \oddb[\dD_L]{K_{A_2}\ltens K_{B_2}}{K_L\ltens K_2}
      \quand
      \odpr[\dD_X]{K_{X}\limp X}
    $$
    for all $X\in\set{A_1,A_2,B_1,B_2}$.
    We conclude by letting $K_A=K_{A_1}\ltens K_{A_2}$ and $K_B=K_{B_1}\ltens K_{B_2}$ since we have a negative derivation of the following form
    $$\small
    \od{\odo{
      \odo{\odo{\odh{
          \odtP{K_A}{\defr}{K_{A_1} \ltens K_{A_2}}{}
          \ltens 
          \odtP{K_B}{\defr}{K_{B_1} \ltens K_{B_2}}{}
        }}{}{
        \oddrP[\dD_R]{K_{A_1}\ltens K_{B_1}}{K_R}
        \ltens
        \oddbP[\dD_L]{K_{A_2}\ltens K_{B_2}}{K_L\ltens K_2}
      }{}}{\tassr}{
        \oddbP[\dD_1]{K_L \ltens K_R }{K_1}
        \ltens K_2
      }{}
    }{\defr}{K}{}}
    $$
    and derivations of the following form with $X\in\set{A,B}$.
    \begin{equation}\label{eq:splitting:1}
      \small
    \od{
      \odo{
        \odi{
          \odh{
            \odprP[\dD_{X_1}]{K_{X_1}\limp X_1}
            \ltens
            \odprP[\dD_{X_2}]{K_{X_2}\limp X_2}
          }
        }{\wtsr}{
          K_{X_1}\limp 
            \odtP{
              X_1\ltens (K_{X_2}\limp X_2)
            }{\curryrule}{
              K_{X_2}\limp (X_1\ltens X_2)
            }{}
        }{\wrule}
      }{\curryrule}{
          \odtP{K_{X_1}\ltens K_{X_2}}{\defr}{K_X}{}
        \limp 
          \odtP{X_1 \ltens X_2}{\defr}{X}{}
      }{}
    }
    \mydot
    \end{equation}

  \item
    %% 2
    The bottom of the derivation is of the following form
    $$\small
      \odt{
        K_1\limp 
          \odnP{
            (A_1\ltens B_1)\ltens((A_2\ltens B_2)\limp K_2)
          }{\wlsr}{
            ((A_1\ltens B_1)\limp(A_2\ltens B_2))\limp K_2
          }{\wrule}
      }{}{
        ((A_1\ltens A_2)\limp(B_1\ltens B_2))\limp (K_1\ltens K_2)
      }{}
    $$
    and we can use the inductive hypothesis to get formulas $K_L$ and $K_R$ such that there are derivations of the following form
    $$\small
      \oddb[\dD_1]{K_L\ltens K_R }{K_1}
    \quand
      \odpr{K_L\limp (A_1\ltens B_1)}
    \quand
      \odt{\odpr{K_R \limp ((A_2\limp B_2)\limp K_2)}}{}{
        (A_2 \limp B_2)\limp (K_R \limp K_2)
      }{}
    $$
    By applying inductive hypothesis on $K_L\limp (A_1\ltens A_2)$ and on $(A_2 \limp B_2)\limp (K_R \limp K_2)$, we found formulas $K_{A_1}$, $K_{A_2}$, $K_{B_1}$ and $K_{B_2}$ such that:
    $$\small
      \oddb[\dD_L]{K_{A_1}\ltens K_{B_1}}{K_L}
      \quand
      \oddr[\dD_R]{K_{A_2}\ltens K_{B_2}}{K_R\ltens K_2}
      \quand
      \odpr[\dD_X]{K_{X}\limp X}
      \quand
      \odpr[\dD_{B_2}]{B_2 \limp K_{B_2}}
    $$
    for all $X\in\set{A_1,A_2,B_1}$.
    We conclude by letting $K_A=K_{A_1}\ltens K_{A_2}$ and $K_B=K_{B_1}\ltens K_{B_2}$, since we have a derivation of the following form
    $$\small
    \od{\odo{
      \odo{\odh{
          \odtP{K_A}{\defr}{\oddb[\dD_L]{K_{A_1}\ltens K_{B_1}}{K_L}}{}
          \limp  
          \odtP{K_B}{\defr}{\oddr[\dD_R]{K_{A_2}\limp K_{B_2}}{K_R\limp K_2}}{}
        }
      }{\curryrule}{
        \oddbP[\dD_1]{K_L \ltens K_R }{K_1}
        \limp K_2
      }{}
    }{\defr}{K}{}}
    $$
    plus a derivation for $K_A \limp A$ as in \eqref{eq:splitting:1}, and a derivation for $B\limp K_B $ as shown below, as required.
    \begin{equation}\label{eq:splitting:2}
      \small
      \od{
        \odo{
          \odo{
            \odh{\odpr[\dD_{B_2}]{B_2\limp K_{B_2}}}
          }{\ludr}{
              \left(\od{
                \odo{
                  \odi{
                    \odh{\oddrP[\dD_{B_1}]{\lunit}{K_{B_1}\limp B_1} 
                    \limp 
                    B_2}
                  }{\blsr}{
                    K_{B_1}\ltens (B_1\limp B_2)
                  }{}  
                }{\tcomr}{
                  (B_1\limp B_2)\ltens K_{B_1}
                }{}
              }\right)
            \limp 
            K_{B_2}
          }{}
        }{\curryrule}{
            \odtP{B_1 \limp B_2}{\defr}{B}{}
          \limp 
            \odtP{K_{B_1}\limp  K_{B_2}}{\defr}{K_B}{}
        }{}
      }
    \end{equation}

    \item
    The bottom of the derivation is of the following form
    $$\small
      \odt{
        \odnP{
          ((K_1\limp (A_1\ltens B_1))\limp (A_2\limp B_2))
        }{\blsr}{
          (K_1\ltens ((A_1\ltens B_1)\limp (A_2\limp B_2))) 
        }{\brule}
        \limp K_2
      }{}{
        (A_1\ltens A_2)\limp(B_1\limp B_2)\limp (K_1\limp K_2)
      }{}
    $$
    and we can use the inductive hypothesis to get formulas $K_L$ and $K_R$ such that there are derivations of the following form
    $$\small
      \oddr[\dD_2]{K_L\limp  K_R }{K_2}
      \quand
      \od{\odo{
        \odp{}{
          K_L \limp (K_1\limp (A_1\ltens B_1))
        }{}
      }{\uncurry}{
        (K_L \ltens K_1)\limp (A_1\ltens B_1)
      }{}
      }
      \quand
      \od{\odp{}{(A_2\limp B_2)\limp K_R}{}}
    $$
    By applying inductive hypothesis on $K_L\limp (K_1\limp (A_1\ltens B_1))$ and on $(A_2 \limp B_2)\limp K_R$, we found formulas $K_{A_1}$, $K_{A_2}$, $K_{B_1}$ and $K_{B_2}$ such that:
    $$\small
      \oddb[\dD_L]{K_{A_1}\ltens K_{B_1}}{K_L\ltens K_1}
      \quand
      \oddr[\dD_R]{K_{A_2}\limp K_{B_2}}{K_R}
      \quand
      \odpr[\dD_X]{K_{X} \limp X}
      \quand
      \odpr[\dD_{B_2}]{B_2 \limp K_{B_2}}
    $$
    for all $X\in\set{A_1,A_2,B_1}$.
    We conclude by letting $K_A=K_{A_1}\ltens K_{A_2}$ and $K_B=K_{B_1}\limp K_{B_2}$, since we have a derivation of the following form
    $$\small
    \od{
      \odo{
        \odi{
          \odo{
            \odo{
              \odh{
                \odtP{K_A}{\defr}{K_{A_1} \ltens K_{A_2}}{}
                \limp
                \odtP{K_B}{\defr}{K_{B_1} \limp K_{B_2}}{}
              }
            }{\uncurry}{
              \odtP{(K_{A_1} \ltens K_{A_2})\ltens K_{B_1}}{}{
                  (K_{A_1} \ltens K_{B_1})\ltens  K_{A_2}
                }{}
              \limp K_{B_2}
            }{}
          }{\curry}{
            \odtP{
              \oddb[\dD_L]{K_{A_1}\ltens K_{B_1}}{K_L\ltens K_1}
            }{\tcomr}{
              K_1\ltens K_L
            }{}
            \limp 
            \oddrP[\dD_R]{K_{A_2}\limp K_{B_2}}{K_R}
          }{}
        }{\wlsr}{
          K_1
          \limp 
          \oddrP[\dD_2]{K_L \ltens K_R }{K_2}
        }{\wrule}
      }{\defr}{K}{}
    }
    $$
    plus a derivation for $K_A \limp A$ of the same form of as the one in \eqref{eq:splitting:1} and a derivation for $B\limp K_B$ as the one in \eqref{eq:splitting:2}, as required.

    \item
    The bottom of the derivation is of the following form
    $$\small
      \odt{
          \odnP{
            (A_1\ltens B_1)\limp ((A_2\limp B_2)\ltens K_2)
          }{\btsr}{
            ((A_1\ltens B_1)\limp (A_2\limp B_2))\ltens K_2
          }{\brule}
        \limp K_1
      }{}{
        ((A_1\ltens A_2)\limp(B_1\limp B_2))\limp (K_2\limp K_1)
      }{}
    $$
    and we can use the inductive hypothesis to get formulas $K_L$ and $K_R$ such that there are derivations of the following form
    $$\small
      \oddr[\dD_1]{K_L\limp  K_R }{K_1}
    \quand
      \odpr{K_L \limp (A_1\ltens B_1)}
    \quand
      \odt{
        \odpr{((A_1\limp B_2)\ltens K_2)\limp K_R}
      }{}{
        (A_2\limp B_2)\limp(K_2\limp K_R)
      }{}
    $$
    By applying inductive hypothesis on $K_L\limp (A_1\ltens B_1)$ and on $(A_2 \limp B_2)\limp K_R$, we find formulas $K_{A_1}$, $K_{A_2}$, $K_{B_1}$ and $K_{B_2}$ such that:
    $$\small
      \oddb[\dD_L]{K_{A_1}\ltens K_{B_1}}{K_L}
      \quand
      \oddr[\dD_R]{K_{A_2}\limp K_{B_2}}{K_2\limp K_R}
      \quand
      \od{\odp{\dD_X}{K_{X} \limp X}{}}
      \quand
      \od{\odp{\dD_{B_2}}{B_2 \limp K_{B_2}}{}}
    $$
    for all $X\in\set{A_1,A_2,B_1}$.
    We conclude by letting $K_A=K_{A_1}\ltens K_{A_2}$ and $K_B=K_{B_1}\limp K_{B_2}$, since we have a derivation of the following form
    $$\small
    \od{
      \odo{
        \odo{
          \odo{
            \odo{
              \odo{
                \odh{
                  \odtP{K_A}{\defr}{K_{A_1} \ltens K_{A_2}}{}
                  \limp
                  \odtP{K_B}{\defr}{K_{B_1} \limp K_{B_2}}{}
                }
              }{\uncurry}{
                \left(\od{
                  \odo{
                    \odo{
                      \odh{(K_{A_1} \ltens K_{A_2})\ltens K_{B_1}}
                    }{\tassr}{
                      K_{A_1} \ltens \odnP{K_{A_2}\ltens K_{B_1}}{\tcomr}{K_{B_1}\ltens K_{A_2}}{}
                    }{}
                  }{\tassr}{
                    (K_{A_1} \ltens K_{B_1})\ltens K_{A_2}
                  }{}
                }\right)
                \limp K_{B_2}
              }{}
            }{\curry}{
              \oddbP[\dD_L]{K_{A_1}\ltens K_{B_1}}{K_L}
              \limp 
              \oddrP[\dD_R]{K_{A_2}\limp K_{B_2}}{K_2\limp K_R}
            }{}
          }{\uncurry}{
            \odnP{K_2\ltens K_L}{\tcomr}{K_L\ltens K_2}{}
            \limp  K_R 
          }{}
        }{\curry}{
          K_2
          \limp
          \oddrP[\dD_1]{K_L\limp  K_R }{K_1}
        }{}
      }{\defr}{K}{}
    }
    $$
    plus a derivation for $K_A \limp A$ of the same form of as the one in \eqref{eq:splitting:1} and a derivation for $B\limp K_B$ as the one in \eqref{eq:splitting:2}, as required.

    \item
    We only discuss the case on the left where $K_2$ and $K_3$ are present.
    If one of these two formulas is absent, then it suffices to consider a different instance of the rule $\wlsq$ (if $K_2$ does not occur), and $\wrsq$ (if $K_3$ does not occur). If both of them do not occur, then it suffices to consider an instance of $\wrr$.
    The case on the right is similar because of associativity of $\lseq$.
    
    Without loss of generality, we assume that the bottom of the derivation is of the following form
    $$\small
      \odt{
        K_1\limp 
          \odnP{
            (K_2\limp A_1)\lseq(K_3\limp( A_3\lseq B))
          }{\qdr}{
            (K_2\lseq K_3)\limp(A_1\lseq(A_2\lseq B))
          }{\wrule}
      }{}{
        (K_1\ltens(K_2\lseq K_3))\limp((A_1\lseq A_2)\lseq B)
      }{}
    $$
    and we can use the inductive hypothesis to get formulas $K_L$ and $K_R$ such that there are derivations of the following form
    $$\small
      \oddb[\dD_1]{K_L\lseq K_R }{K_1}
    \quand
      \odt{
        \odpr[\dD_{A_1}']{K_L \limp (K_2\limp A_1)}
      }{}{
        (K_L \ltens K_2)\limp A_1
      }{}
    \quand
      \odt{
        \odpr{K_R \limp (K_3\limp (A_2\lseq B))}
      }{}{
        (K_R \ltens K_3)\limp (A_2\lseq B)
      }{}
    $$
    By applying inductive hypothesis on $(K_R \ltens K_3)\limp (A_2\lseq B)$, we find formulas $K_{A_2}$, $K_{B}$ such that:
    $$\small
      \oddb[\dD_R]{K_{A_2}\lseq K_{B}}{K_R\ltens K_3}
      \quand
      \odpr[\dD_{A_2}]{K_{A_2} \limp A_2}
      \quand
      \odpr[\dD_{B}]{K_B\limp B}
    $$
    We conclude by letting $K_A=(K_L \ltens K_2)\lseq K_{A_2}$, since we have a derivation for $K_B\limp B$ shown above, and the two derivations below: 
    $$
    \small
    \od{
      \odi{
        \odo{\odh{
          ((K_L\ltens K_2)\lseq K_{A_2})\lseq K_B
        }}{\sassr}{
          (K_L\ltens K_2) 
          \lseq 
          \oddbP[\dD_R]{K_{A_2}\lseq K_{B}}{K_R\ltens K_3}
        }{}
      }{\bqdr}{
        \oddbP[\dD_1]{K_L\lseq K_R }{K_1}\ltens (K_2\lseq K_3)
      }{\brule}
    }
    \qquad
    \od{
      \odi{
        \odi{\odh{\lunit}}{\wsudr}{
          \odtP{
            \oddr[\dD_{A_1}']{\lunit}{K_L \limp (K_2\limp A_1)}
          }{\uncurry}{
            (K_L \ltens K_2)\limp A_1
          }{}
          \lseq
          \oddrP[\dD_{A_2}]{\lunit}{K_{A_2} \limp A_2}
        }{}
      }{\wqdr}{
        ((K_L \ltens K_2)\lseq K_{A_2})\limp(A_1\lseq A_2)
      }{\wrule}
    }
    $$
    as required.

    \item
    As in the previous case, we only discuss the case on the left in which $K_1$ and $K_2$ are present.
    If one of these two formulas is absent, then it suffices to consider a different instance of the rule $\blsq$ (if $K_2$ does not occur), or $\brsq$ (if $K_3$ does not occur). If both of them do not occur, then it suffices to consider an instance of $\brr$.
    The case on the right is similar because of associativity of $\lseq$.

    We can assume that the bottom of the derivation is of the following form
    $$\small
      \odt{
          \odnP{
            (A_1\ltens K_1)\lseq ((A_2\lseq B)\ltens K_2)
          }{\qdr}{
            (A_1\lseq(A_2\lseq B))\limp( K_1\lseq K_2)
          }{\brule}
        \limp K_3
      }{}{
        ((A_1\lseq A_2)\lseq B)\limp ((K_1\lseq K_2)\limp K_3)
      }{}
    $$
    and we can use the inductive hypothesis to get formulas $K_L$ and $K_R$ such that there are derivations of the following form
    $$\small
      \oddr[\dD_1]{K_L\lseq K_R }{K_3}
    \quand
      \odt{
        \odpr[\dD_{A_1}']{(A_1\ltens K_1)\limp K_L}
      }{\curry}{
        A_1\limp(K_1\limp K_L)
      }{}
    \quand
      \odt{
        \odpr{((A_2\lseq B) \ltens K_2)\limp K_R}
      }{\curry}{
        (A_2\lseq B)\limp (K_2\limp K_R)
      }{}
    $$
    By applying inductive hypothesis on $(A_2\lseq B)\limp (K_2\limp K_R)$, we find formulas $K_{A_2}$ and $K_{B}$ such that:
    $$\small
      \oddb[\dD_R]{K_{A_2}\lseq K_{B}}{K_2\limp K_R}
      \quand
      \odpr[\dD_{A_2}]{A_2\limp K_{A_2}}
      \quand
      \odpr[\dD_{B}]{B \limp K_B}
      \mydot
    $$
    We conclude by letting $K_A=(K_{A_1}\limp K_{L})\lseq K_{A_2}$, since we have a derivation for $K_B\limp B$ shown above, and derivations of the following form
    $$\small
    \od{
      \odo{
        \odi{
          \odo{
            \odh{
              \odtP{K_A}{\defr}{(K_{A_1}\limp K_{L})\lseq K_{A_2}}{}
              \lseq K_B
            }
          }{\sassl}{
            (K_1\limp K_L)
            \lseq
            \oddbP[\dD_R]{K_{A_2}\lseq K_{B}}{K_2\limp K_R}
          }{}
        }{\wqdr}{
          (K_1\lseq K_2)
          \limp
          \oddrP[\dD_1]{K_L\lseq K_R }{K_3}
        }{\wrule}
      }{\defr}{K}{}
    }
    \qquad
      \od{\odi{
        \odi{\odh{\lunit}}{\wsudr}{
          \odtP{
            \oddr[\dD_{A_1}']{\lunit}{(A_1\ltens K_1)\limp K_L}
          }{\curry}{
            A_1\limp(K_1\limp K_L)
          }{}
          \lseq
          \oddrP[\dD_{A_2}]{\lunit}{A_2\limp K_{A_2}}
        }{\wrule}
      }{\wqdr}{
        \odtP{A_1\lseq A_2}{\defr}{A}{}
        \limp 
        \odtP{(K_{A_1}\limp K_{L})\lseq K_{A_2}}{\defr}{K_A}{}
      }{\wrule}
    }
    $$
    as required.

    \item
    As in the previous cases, we only discuss the case on the left in which $K_1$ and $K_2$ are present.
    If one of these two formulas does not occur, then it suffices to consider a different instance of $\wlsq$ (if $K_2$ does not occur), and $\wrsq$ (if $K_3$ does not occur). If both of them do not occur, then it suffices to consider an instance of $\wrr$.
    The rightmost case is similar because of associativity of $\lseq$.

    We can assume that the bottom of the derivation is of the following form
    $$\small
      \odt{
        K_1\limp 
          \odnP{
            (A_1\limp K_2)\lseq((A_2\lseq B)\limp K_3)
          }{\qdr}{
            ((A_1\lseq A_2)\lseq B) \limp (K_2\lseq K_3)
          }{\wrule}
      }{}{
        ((A_1\lseq A_2)\lseq B) \limp (K_1\limp  (K_2\lseq K_3))
      }{}
    $$
    and we can use the inductive hypothesis to get formulas $K_L$ and $K_R$ such that there are derivations of the following form
    $$\small
      \oddb[\dD_1]{K_L\lseq K_R }{K_1}
    \quand
      \odt{
        \odpr[\dD_{A}']{K_L \limp (A_1\limp K_2)}
      }{}{
        A_1\limp (K_L\limp K_2)
      }{}
    \quand
      \odt{
        \odpr{K_R\limp((A_2\lseq B)\limp K_3)}
      }{}{
        (A_2\lseq B)\limp (K_R \limp K_3)
      }{}
      \mydot
    $$
    By applying inductive hypothesis on $(A_2\lseq B)\limp (K_R \limp K_3)$, we have formulas $K_{A_2}$ and $K_{B}$ such that:
    $$\small
      \oddr[\dD_R]{K_{A_2}\lseq K_{B}}{K_R\limp K_3}
      \quand
      \odpr[\dD_{A_2}]{A_2\limp K_{A_2}}
      \quand
      \odpr[\dD_{B}]{B \limp K_B}
      \mydot
    $$
    We conclude by letting $K_A=(K_L\limp K_2)\lseq K_{A_2}$, since we have a derivation for $K_B\limp B$ shown above, and derivations of the following form
    $$\small
    \od{
      \odi{
        \odo{\odh{
          \odtP{K_A}{\defr}{(K_L\limp K_2)\lseq K_{A_2}}{}\lseq K_B
        }}{\sassl}{
          (K_L\limp K_2)
          \lseq
          \oddrP[\dD_R]{K_{A_2}\lseq K_{B}}{K_R\limp K_3}
        }{}
      }{\wqdr}{
        \oddbP[\dD_1]{K_L\lseq K_R }{K_1}\limp (K_2\lseq K_3)
      }{\wrule}
    }
    \quad
    \od{
      \odi{
        \odi{\odh{\lunit}}{\wsudr}{
          \left(\od{
            \odo{
              \odo{
                \odh{
                  \oddr[\dD_{A}']{\lunit}{K_L \limp (A_1\limp K_2)}  
                }
              }{\uncurry}{
                \odtP{K_L \ltens A_1}{\tcomr}{
                  A_1 \ltens K_L
                }{}
                \limp K_2
              }{}
            }{\curry}{
              A_1 \limp (K_L\limp K_2)
            }{}
          }\right)
          \lseq
          \oddrP[\dD_{A_2}]{\lunit}{A_2\limp K_{A_2}}
        }{\wrule}
      }{\bqdr}{
        \odnP{A_1\lseq A_2}{}{A}{}
        \limp 
        \odnP{(K_L\limp K_2)\lseq K_{A_2}}{}{K_A}{}
      }{\brule}
    }
    $$
    as required.
  \end{enumerate}
\end{proof}
%%%%%%%%%%%%%%%%%%%%%%%%%%%%%%%%%%%%%%%%%%%%%%%%%%%%%%%%%%%%%%%%
%%%%%%%%%%%%%%%%%%%%%%%%%%%%%%%%%%%%%%%%%%%%%%%%%%%%%%%%%%%%%%%%

%%%%%%%%%%%%%%%%%%%%%%%%%%%%%%%%%%%%%%%%%%%%%%%%%%%%%%%%%%%%%%%%
\contextReduction*
%%%%%%%%%%%%%%%%%%%%%%%%%%%%%%%%%%%%%%%%%%%%%%%%%%%%%%%%%%%%%%%%
\begin{proof}
  By induction on the structure of each context --- see \eqref{eq:contexts}.
  We here only discuss here the cases for the positive contexts, since the arguments for the same cases for the negative contexts are analogous, with some minor differences in the (polarities of the) rules used.
  \begin{itemize}
    \item 
    If $\cP=H\limp \ctx $ we conclude immediately by letting $H=K$ and $\dD_X$ being the identity derivation.
    
    \item 
    If $\cP=H\limp(J\ltens \cPp)$, then we can apply splitting to get derivations 
    $$\small
      \oddb[\dD_H]{H_J\ltens H_P}H
    \quand
      \odpr[\dD_J]{H_J\limp J}
    \quand
      \odpr{H_P\limp \cPp[A]}
    $$
    We can then apply the inductive hypothesis on $H_P\limp \cPp[A]$, giving us derivations
    \begin{equation}\label{eq:contextReduction:1}\small
      \oddr[\dD_P]{K\limp X}{H_P \limp \cPp[X]}
    \quand
      \odpr[\dD_A]{K\limp A}
    \qquad \mbox{for any formula $X$.}
    \end{equation}
    We conclude since we have the following derivation for any formula $X$:
    $$\small
    \od{
      \odo{
        \odo{
          \odh{\oddr[\dD_P]{K\limp X}{
            H_P 
            \limp 
            \odnP{\oddrP[\dD_J]{\lunit}{H_J\limp J}\ltens \cPp[X]
            }{\wtsr}{
              H_J\limp (J \ltens \cPp[X])
            }{\wrule}
          }}
        }{\uncurry}{
          \oddbP[\dD_H]{H_J\ltens H_P}H
          \limp 
          (J\ltens \cPp[X])
        }{}
      }{\defr}{
        \cP[X]
      }{}
    }
    $$
    
    \item 
    The case $\cP=H\limp(\cPp\ltens J)$ is similar to the previous one.

    \item
    If $\cP=H\limp (J\lseq \cPp)$, then we can apply splitting to get derivations 
    $$\small
      \oddr[\dD_H]{H_J\lseq H_P}H
    \quand
      \odpr[\dD_J]{H_J\limp J}
    \quand  
      \odpr{H_P\limp \cPp[A]}
    $$
    and we can apply the inductive hypothesis on $H_P\limp \cPp[A]$, giving us derivations as in \Cref{eq:contextReduction:1}.
    We conclude since we have the following derivation for any formula $X$:
    $$\small
      \od{
        \odo{
          \odi{\odi{\odh{\oddr[\dD_P]{K\limp X}{H_P \limp \cPp[X]}}
            }{\wsudr}{
              \oddrP[\dD_J]{\lunit}{H_J\limp J} \lseq (H_P \limp \cPp[X])
            }{\wrule}
          }{\wqdr}{
            \oddbP[\dD_H]{H_J\lseq H_P}H
            \limp 
            (J\lseq \cPp[X])
          }{\wrule}
        }{\defr}{
          \cP[X]
        }{}
      }
    $$

    \item 
    The case $\cP=H\limp (\cPp\lseq  J)$ is similar to the previous one.

    \item 
    If $\cP=(\cPp\limp J)\limp H$, then we can apply splitting to get derivations
    $$\small
      \oddr[\dD_H]{H_P\limp H_J}H
    \quand
      \odpr[\dD_J]{J\limp H_J}
    \quand
      \odpr{H_P\limp \cPp[A]}
    $$
    and we can apply the inductive hypothesis on $H_P\limp \cPp[A]$, giving us derivations as in \eqref{eq:contextReduction:1}.
    We conclude since we have the following derivation for any formula $X$:
    $$\small
      \od{
        \odo{
          \odo{
            \odo{
              \odh{
                \oddr[\dD_P]{K\limp X}{
                  H_P 
                  \limp 
                  \left(\od{
                    \odi{
                      \odo{\odh{\cPp[X]}
                      }{\tudr}{
                        \cPp[X]\ltens \oddrP[\dD_J]\lunit{J\limp H_J}
                      }{}
                    }{\wlsr}{
                      (\cPp[X]\limp J)\limp H_J
                    }{\wrule}
                  }\right)
                }
              }
            }{\uncurry}{
              \odtP{
                H_P \ltens (\cPp[X]\limp J)
              }{\tcomr}{
                (\cPp[X]\limp J)\ltens H_P
              }{}
              \limp H_J
            }{}
          }{\curry}{
            (\cPp[X]\limp J)\limp \oddrP[\dD_H]{H_P\limp H_J}{H}
          }{}
        }{\defr}{
          \cP[X]
        }{}
      }
    $$

    \item 
    If $\cP=(J\limp \cN)\limp H$, then we can apply splitting to get derivations
    $$\small
      \oddr[\dD_H]{H_J\limp H_P}{H}
    \quand
      \odpr[\dD_J]{H_J\limp J}
    \quand
      \odpr{\cN[A]\limp H_P}
    $$
    and we can apply the inductive hypothesis on $\cN[A]\limp H_P$, giving us the following derivations:
    \begin{equation}\label{eq:contextReduction:2}
      \small
      \oddr[\dD_P]{K\limp X}{\cN[X]\limp H_P}
    \quand
      \odpr[\dD_A]{K\limp A}
    \qquad \mbox{for any formula $X$.}
    \end{equation}
    We conclude since we have the following derivation for any formula $X$:
    $$\small
    \oddr[\dD_P]{K\limp A}{
      \od{
        \odo{
          \odo{
            \odh{
              \left(\od{
                \odo{
                  \odi{
                    \odo{\odh{\cN[X]}}{\ludr}{
                    \oddrP[\dD_J]{\lunit}{H_J\limp J}
                    \limp
                    \cN[X]
                    }{}
                  }{\blsr}{
                    H_J\ltens (J\limp \cN[X])
                  }{\brule}
                }{\tcomr}{
                  (J\limp \cN[X])\ltens H_J
                }{}
              }\right)
              \limp 
              H_P
            }
          }{\curry}{
            (J\limp \cN[X])
            \limp 
            \oddrP[\dD_H]{H_J\limp H_P}H
          }{}
        }{\defr}{
          \cP[X]
        }{}
      }
    }
    $$

    \item 
    If $\cP=(\cN\lseq J)\limp H$, then we can apply splitting to get derivations
    $$\small
      \oddr[\dD_H]{H_P\lseq H_J}H
    \quand
      \odpr[\dD_J]{J\limp H_J}
    \quand
      \odpr{\cN[A]\limp H_P}
    $$
    and we can apply the inductive hypothesis on $\cN[A]\limp H_P$, giving us derivations as in \eqref{eq:contextReduction:2}.
    We conclude since we have the following derivation for any formula $X$:
    $$\small
      \od{
        \odo{
          \odi{\odi{\odh{X\limp K}}{\wcudr}{
              \oddrP[\dD_P]{X\limp K}{\cN[X]\limp H_P}
              \lseq 
              \oddrP[\dD_J]{\lunit}{J\limp H_J}
            }{\wrule}
          }{\wqdr}{
            (\cN[X]\lseq J)
            \limp 
            \oddrP[\dD_H]{H_P\lseq H_J}H
          }{\wrule}
        }{\defr}{
          \cP[X]
        }{}
      }
    $$

    \item 
      The case $\cP=(J\lseq \cN)\limp H$ is similar to the previous one.
      \qed
  \end{itemize}
\end{proof}
%%%%%%%%%%%%%%%%%%%%%%%%%%%%%%%%%%%%%%%%%%%%%%%%%%%%%%%%%%%%%%%%
%%%%%%%%%%%%%%%%%%%%%%%%%%%%%%%%%%%%%%%%%%%%%%%%%%%%%%%%%%%%%%%%

%%%%%%%%%%%%%%%%%%%%%%%%%%%%%%%%%%%%%%%%%%%%%%%%%%%%%%%%%%%%%%%%
\uprules*
%%%%%%%%%%%%%%%%%%%%%%%%%%%%%%%%%%%%%%%%%%%%%%%%%%%%%%%%%%%%%%%%
\begin{proof}
  We provide a procedure to remove the up rules from a derivation in $\SBiV$ proceeding top-down, showing that whenever there is a derivation of the conclusion of the up rule, then there is a derivation of its premise.
  \begin{itemize}
    \item 
    In the case of $\rrule_\uparrow^\bullet\in \set{\tuur,\luur,\bsudr,\bcudr}$,
    we only show the case $\tuur$, since the other are similar.
    We have a derivation of the form {\small$\odpr{\cN[\odn{\lunit \ltens A}{\rrule_\uparrow^\bullet}{A}{\brule}]}$} and we apply \Cref{lemma:contextReduction} to $\cN[\lunit \ltens A]$ to get derivations
    \begin{equation}
      \small
      \oddr[\dD_X]{X\limp K}{\cN[X]}
      \quand
      \odpr[\dD']{(\lunit\ltens A)\limp K}
    \end{equation}
    for any formula $X$.
    The Splitting Lemma~\ref{lemma:splitting} applied to $(\lunit\ltens A)\limp K$ gives us derivations: 
    $$\small
    \oddr[\dD_K]{\lunit \ltens K_A}{K}
    \quand
    \odpr[\dD_{K_A}]{A\limp K_A}
    \quand
    \oodh{\lunit}
    \mydot
    $$
    We conclude by constructing the following derivation:
    $$\small
    \oddr[\dD_A]{
      \odpr[\dD_{K_A}]{
        A 
        \limp 
        \odtP{K_A}{\tudr}{\oddr[\dD_{K}]{\lunit\ltens K_A}{K}}{}
      }
    }{\cN[A]}
    $$

    \ 

    \item 
    In the case of $\baiur$, we have a derivation of the form {\small$\odpr{\cN[\odn{a\limp a}{\baiur}{\lunit}{\brule}]}$}.
    We apply \Cref{lemma:contextReduction} to $\cN[a\limp a]$ to get a formula $K$ such that there are derivations
    \begin{equation}\label{eq:upElim:1}
      \small
      \oddr[\dD_X]{X\limp K}{\cN[X]}
      \quand
      \odpr[\dD_a]{(a\limp a)\limp K}
    \end{equation}
    for any formula $X$.
    The Splitting Lemma~\ref{lemma:splitting} applied to $(a\limp a)\limp K$ gives us derivations: 
    $$\small
    \oddr[\dD_K]{K_L\limp K_R}{K}
    \quand
    \odpr{K_L\limp a}
    \quand
    \odpr{a\limp K_R}
    \mydot
    $$
    To which we apply the Atomic Splitting Lemma~\ref{lemma:atomicSplitting}, and we can conclude by constructing the following derivation:
    $$\small
    \od{
      \odo{
        \odi{\odh{\lunit}}{\waidr}{
          \oddr[\dD_K]{
            \oddbP[\dD_L]{a}{K_L}
            \limp 
            \oddrP[\dD_R]{a}{K_R}
          }{K}
        }{\wrule}
      }{\ludr}{
        \oddr[\dD_\lunit]{\lunit\limp K}{\cN[\lunit]}
      }{}
    }
    $$
    where 
    $\dD_L$ and $\dD_R$ are obtained applying \Cref{lemma:atomicSplitting} to $K_L\limp a$ and $a\limp K_R$, and where $\dD_\lunit$ is the instantiation of the derivation $\dD_X$ form with $X=\lunit$.
    
    \item 
    In the case of $\wqur$ applied in a positive context, we have a derivation of the form 
    {\small $\odpr{\cP[\odn{(A\lseq C)\ltens(B\lseq D)}{\wqur}{(A\ltens B)\lseq (C\ltens D)}{\wrule}]}$}.
    We apply \Cref{lemma:contextReduction} to the conclusion to get a formula $K$ and the following derivations:
    \begin{equation}\label{eq:upElim:3}
      \oddr[\dD_X]{K\limp X}{\cP[X]}
      \quand
      \odpr[\dD_a]{((A\ltens B)\lseq (C\ltens D))\limp K}
      \mydot
    \end{equation}
    The fact that $\proves[\BiV](A\lseq C)\ltens(B\lseq D)\limp K$ allows us to apply the Splitting \Cref{lemma:splitting} to get the following derivations:
    $$
    \oddb[\dD_K]{K_L\ltens K_R}{K}
    \quand
    \odpr{K_L\limp (A\lseq C)}
    \quand
    \odpr{K_R\limp (B\lseq D)}
    \mydot
    $$
    We conclude because of the existence of the following derivation:
    $$\small
    \od{
      \odd{
        \odo{\odi{
          \odo{\odh{\lunit}}{\wsudr}{
            \oddrP[\dD_A]{\lunit}{K_A\limp A}\lseq \oddrP[\dD_C]{\lunit}{K_C\limp C}
          }{}
        }{\wqdr}{
          (K_A\lseq K_C) 
          \limp
            \odnP{
                \left(\od{\odi{\odo{\odh A}{\tudr}{\oddrP[\dD_B]{\lunit}{K_B \limp B} \ltens A}{}
                  }{\wlsr}{
                    K_B \limp (B\ltens A)
                  }{\wrule}
                }\right)
              \lseq 
                \left(\od{\odi{\odo{\odh C}{\tudr}{\oddrP[\dD_D]{\lunit}{K_D \limp D} \ltens C}{}
                  }{\wlsr}{
                    K_D \limp (D\ltens C)
                  }{\wrule}
                }\right)
            }{\wqdr}{
              (K_B\lseq K_D)\limp ((B\ltens A)\lseq(D\ltens C))
            }{\wrule}
        }{\wrule}}{\curryrule}{
          \oddbP[\dD_K]{
            \oddbP{K_A\lseq K_C}{K_L}
            \ltens 
            \oddbP{K_B\lseq K_D}{K_R}
          }{K}
          \limp 
          \left(
            \odtP{B\ltens A}{\defr}{A \ltens B}{}
            \lseq
            \odtP{D\ltens C}{\defr}{C \ltens D}{}
          \right)
        }{}
      }{\dDp_{(A\ltens B)\lseq(C\ltens D)}}{
        \cP[(A\ltens B)\lseq(C\ltens D)]
      }{}
    }
    $$
    where the existence of derivations 
    $\dD_A$, $\dD_C$, and $\dD_L$ \resp{$\dD_B$, $\dD_D$, and $\dD_R$}
    are guaranteed by \Cref{lemma:splitting} applied to 
    $K_L\limp (A\lseq C)$ \resp{$(B\lseq D)\limp K_R$}, 
    and where $\dD_P$ is the instantiation of the derivation $\dD_X$ with $P=(A\lseq C)\ltens(B\lseq D)$.

    \item 
    In the case of $\bqur$ applied in a negative context, we have a (positive) derivation of the form 
    {\small $\odpr{\cN[\odn{(A\lseq C)\limp(B\lseq D)}{\bqur}{(A \limp B) \lseq (C\limp D)}{\brule}]}$}.
    We apply \Cref{lemma:contextReduction} to the conclusion to get a formula $K$ and the following derivations:
    \begin{equation}\label{eq:upElim:2}
      \oddb[\dD_X]{X\limp K}{\cN[X]}
      \quand
      \odpr[\dD_a]{((A \limp B) \lseq (C\limp D))\limp K}
      \mydot
    \end{equation}
    The fact that $\proves[\BiV](A\lseq C)\limp(B\lseq D)\limp K$ allows us to apply the splitting lemma \Cref{lemma:splitting} to get the following derivations:
    $$
    \oddr[\dD_K]{K_L\limp K_R}{K}
    \quand
    \odpr{K_L\limp (A\lseq C)}
    \quand
    \odpr{(B\lseq D)\limp K_R}
    \mydot
    $$
    We conclude because of the existence of the following derivation:
    $$\small
    \od{
      \odd{
        \odo{\odi{
          \odi{\odh{\lunit}}{\wsudr}{
            \oddrP[\dD_B]{\lunit}{B\limp K_B}\lseq \oddrP[\dD_D]{\lunit}{D\limp K_D}
          }{\wrule}
        }{\wqdr}{
          \odnP{
            \left(\od{
              \odi{\odo{\odh B}{\ludr}{\oddrP[\dD_A]{\lunit}{K_A\limp A} \limp B}{}}{\wlsr}{(A\limp B) \ltens K_A}{\wrule}
            }\right)
            \lseq 
            \left(\od{
              \odi{\odo{\odh D}{\ludr}{\oddrP[\dD_C]{\lunit}{K_C\limp C} \limp D}{}}{\wlsr}{(C\limp D) \ltens K_C}{\wrule}
            }\right)
          }{\bqdr}{
            ((A\limp B)\lseq(C\limp D))\ltens(K_A\lseq K_C)
          }{\brule}
          \limp
          (K_B\lseq K_D)
        }{\wrule}}{\curry}{
          ((A\limp B)\lseq(C\limp D))
          \limp 
          \oddrP[\dD_K]{
            \oddrP[\dD_L]{K_A\lseq K_C}{K_L}
            \limp 
            \oddrP[\dD_R]{K_B\lseq K_D}{K_R}
          }{K}
        }{}
      }{\dDn_{(A\limp B)\lseq (C\limp D)}}{
        \cN[(A\limp B)\lseq (C\limp D)]
      }{}
    }
    $$
    where the existence of derivations 
    $\dD_A$, $\dD_C$, and $\dD_L$ \resp{$\dD_B$, $\dD_D$, and $\dD_R$}
    are guaranteed by \Cref{lemma:splitting} applied to 
    $K_L\limp (A\lseq C)$ \resp{$(B\lseq D)\limp K_R$}, 
    and where $\dD_P$ is the instantiation of the derivation $\dD_X$ with $P=(A\lseq C)\limp(B\lseq D)$.
\qed
    
  \end{itemize}
\end{proof}
%%%%%%%%%%%%%%%%%%%%%%%%%%%%%%%%%%%%%%%%%%%%%%%%%%%%%%%%%%%%%%%%
%%%%%%%%%%%%%%%%%%%%%%%%%%%%%%%%%%%%%%%%%%%%%%%%%%%%%%%%%%%%%%%%

%%%%%%%%%%%%%%%%%%%%%%%%%%%%%%%%%%%%%%%%%%%%%%%%%%%%%%%%%%%%%%%%
%%%%%%%%%%%%%%%%%%%%%%%%%%%%%%%%%%%%%%%%%%%%%%%%%%%%%%%%%%%%%%%%
\section{Details for  \Cref{sec:assoc}}
\label{app:assoc}
%%%%%%%%%%%%%%%%%%%%%%%%%%%%%%%%%%%%%%%%%%%%%%%%%%%%%%%%%%%%%%%%
%%%%%%%%%%%%%%%%%%%%%%%%%%%%%%%%%%%%%%%%%%%%%%%%%%%%%%%%%%%%%%%%

\BiVextensNiML*
\begin{proof}
  For the $\qdr$-rules, we have the following derivations.
  $$\adjustbox{max width=\textwidth}{$
  \vlderivation{
    \vlin{\rlimp}{}{
      (A\limp B)\lseq(C\limp D)\vdash(A\lseq C)\limp (B\lseq D)
    }{
      \vliin{\lseq}{}{
        (A\limp B)\lseq(C\limp D),A\lseq C\vdash B\lseq D
      }{
        \vliin{\llimp}{}{
          A\limp B, A\vdash B
        }{
          \vlin{\AXrule}{}{A\vdash A}{\vlhy{}}
        }{
          \vlin{\AXrule}{}{B\vdash B}{\vlhy{}}
        }
      }{
        \vliin{\llimp}{}{
          C\limp D, C\vdash D
        }{
          \vlin{\AXrule}{}{C\vdash C}{\vlhy{}}
        }{
          \vlin{\AXrule}{}{D\vdash D}{\vlhy{}}
        }
      }
    }
  }
  \qquad
  \vlderivation{
    \vlin{\lltens}{}{
      (A\lseq C)\ltens (B\lseq D)\vdash (A\ltens B)\lseq(C\ltens D)
    }{
      \vliin{\lseq}{}{
        A\lseq C,B\lseq D\vdash (A\ltens B)\lseq(C\ltens D)
      }{
        \vliin{\rltens}{}{
          A, B\vdash A\ltens B
        }{
          \vlin{\AXrule}{}{A\vdash A}{\vlhy{}}
        }{
          \vlin{\AXrule}{}{B\vdash B}{\vlhy{}}
        }
      }{
        \vliin{\rltens}{}{
          C, D\vdash C\ltens D
        }{
          \vlin{\AXrule}{}{C\vdash C}{\vlhy{}}
        }{
          \vlin{\AXrule}{}{D\vdash D}{\vlhy{}}
        }
      }
    }
  }
  $}$$
  For the $\sqir_{\mathsf{L}}$-rules, we have the following derivations:
  $$\small
    \vlderivation{
      \vlin{\rlimp}{}{
        (A\limp B)\lseq C \vdash A\limp(B\lseq C)
      }{
        \vliin{\lseq}{}{(A\limp B)\lseq C , A\vdash B\lseq C}{
          \vliin{\llimp}{}{
            A\limp B, A\vdash B
          }{
            \vlin{\AXrule}{}{A\vdash A}{\vlhy{}}
          }{
            \vlin{\AXrule}{}{B\vdash B}{\vlhy{}}
          }
        }{
          \vlin{\AXrule}{}{C\vdash C}{\vlhy{}}
        }
      }
    }
  \qquad\qquad
    \vlderivation{
      \vlin{\lltens}{}{
        A\ltens(B\lseq C) \vdash (A\ltens B)\lseq C
      }{
        \vliin{\lseq}{}{
          A,B\lseq C \vdash (A\ltens B)\lseq C
        }{
          \vliin{\rltens}{}{
            A,B \vdash A\ltens B
          }{
            \vlin{\AXrule}{}{A\vdash A}{\vlhy{}}
          }{
            \vlin{\AXrule}{}{B\vdash B}{\vlhy{}}
          }
        }{
          \vlin{\AXrule}{}{C\vdash C}{\vlhy{}}
        }
      }
    }
  $$
  Similar derivations can be given for the $\sqir_{\mathsf{R}}$-rules.
  Finally, we have the following derivation for the $\refrule$-rules.
  $$\small
  \vlderivation{
    \vlin{\lltens}{}{
      A \ltens B \vdash A \lseq B
    }{
      \vliin{\lseq}{}{
        A, B \vdash A \lseq B
      }{
        \vlin{\AXrule}{}{A\vdash A}{\vlhy{}}
      }{
        \vlin{\AXrule}{}{B\vdash B}{\vlhy{}}
      }
    }
  }
  $$
  Conversely, after \Cref{lemma:IMLLtoIBV}, it suffices to prove that if all premises of the rule $\lseq$ are derivable in $\IBV$, then we can derive the conclusion.
  This is shown using the following derivation.
  $$
  \adjustbox{max width=\textwidth}{$
    \od{
    \odo{
      \odi{\odh{
        \odtP{
          \odtP{
            \Gamma^{\ltens} \ltens A_1\ltens\cdots\ltens A_n
          }{\tcomr}{
            A_1\ltens\cdots\ltens A_n\ltens \Gamma^{\ltens}
          }{}
          \limp A
        }{\uncurry}{
          (A_1\ltens\cdots\ltens A_n)\limp (\Gamma^{\ltens} \limp A)
        }{}
        \lseq
        \odtP{
          \odtP{
            \Delta^{\ltens} \ltens B_1\ltens\cdots\ltens B_n
          }{\tcomr}{
            B_1\ltens\cdots\ltens B_n\ltens \Delta^{\ltens} 
          }{}
          \limp A
        }{\uncurry}{
          (B_1\ltens\cdots\ltens B_n)\limp (\Delta^{\ltens} \limp B)
        }{}
      }}{\wqdr}{
        \left(\od{
          \odi{\odh{(A_1\ltens\cdots\ltens A_n)\lseq (B_1\ltens\cdots\ltens B_n)}}{(n-1)\times\bqdr}{
            (A_1\lseq B_1)\ltens\cdots\ltens (A_n\lseq B_n)
          }{\brule}
        }\right)
        \limp 
        \odnP{
          (\Gamma^{\ltens} \limp A)\lseq (\Delta^{\ltens} \limp B)
        }{\wqdr}{
          \odnP{\Gamma^{\ltens} \lseq \Delta^{\ltens}}{\brr}{
            \Gamma^{\ltens} \ltens \Delta^{\ltens}
          }{\brule}
          \limp (A\lseq B)
        }{\wrule}
      }{\wrule}
    }{\curry}{
      \odtP{
        (A_1\lseq B_1)\ltens\cdots\ltens (A_n\lseq B_n)\ltens (\Gamma^{\ltens} \ltens \Delta^{\ltens})
      }{n\times \tcomr}{
        \Gamma^{\ltens} \ltens \Delta^{\ltens}
        \ltens(
        A_1\lseq B_1)\ltens\cdots\ltens (A_n\lseq B_n)
      }{}
      \limp 
      (A\lseq B)
    }{}
    }
    $}
  $$
  
\end{proof}

\clearpage

\NMLextensNiML*
\begin{proof}
  First, let $\pi_A$ be a proof of $A$ in $\INML$. We construct a proof $\emb\pi_A$ of $\emb A$ in $\NML$ by translating each sequent $B_1,\ldots B_n\vdash C$ in $\pi_A$ into the sequent $\vdash(\emb{B_1})\lneg,\ldots,(\emb{B_n})\lneg,\emb{C}$.
  In fact, this translation automatically translates every instance of a rule $\rrule\in\INML$ in $\pi_A$ into a (correct) instance of a rule $\emb{\rrule}\in\NML$ --- see \eqref{eq:precRule} and Figures~\ref{fig:IMLL} and~\ref{fig:NML}. Note that $A$ is unit-free, and therefore the two rules $\lunitl$ and $\lunitr$ do not occur in~$\pi_A$.

  Conversely, consider an arbitrary (cut-free) proof $\pi$ in $\NML$.
  Every axiom in $\pi$ contains one positive and one negative formula.
  Each inference rule in $\pi$ that has one premise either keeps the number of positive formulas in the sequent constant, or increases the number of positive formulas in the sequent, or introduces an unpolarisable formula in the sequent.
  For each inference rule in $\pi$ with two premises, either the number of positive formulas in the conclusion is equal or bigger than the maximum of the number of positive formulas in the premises, or the conclusion contains an unpolarisable formula.
  And if a sequent in $\pi$ contains an unpolarisable formula, then so does every sequent below.

  This means that whenever the conclusion of $\pi$ contains exactly one positive formula and no unpolarisable formulas, then so does every line in $\pi$. Then, every $\NML$-sequent in $\pi$ has the shape $\vdash\bB_1,\ldots,\bB_n,\wC$. By Lemma~\ref{lem:emb}, we can transform this into an $\INML$-sequent $B'_1,\ldots,B'_n\vdash C'$, such that $\emb{(B_i')}=(\bB_i)\lneg$ for $i=1,\ldots, n$ and $\emb{(C')}=\wC$. Furthermore, under this transformation, every rule instance $\rrule$ in $\pi$ becomes a valid rule instance $\rrule^\sharp$ in $\INML$.\footnote{This also proves that $\MLLx$ is a conservative extension of $\IMLL$. This should be folklore knowledge, but we could not find a proof of this in the literature.}
  \qed
\end{proof}

From the above proof, we immediately have the following result, which should be folklore knowledge, but we could not find a proof of this in the literature.
\begin{corollary}
  Unit-free $\MLLx$ is a conservative extension of unit-free $\IMLL$. 
\end{corollary}

\BiVextensNiMLcut*
\begin{proof}
  For the left-to-right direction, we can use \Cref{thm:BiVextensNiML}, and it suffices to show that the rules $\sassl$ and $\sassr$ are derivable in $\NiML\cup\set{\acutl,\acutr}$.
  For $\sassr$ we have the following derivation.
  $$
    \vlderivation{
        \vliin{\acutr}{}{
          (A\lprec B)\lprec C  \vdash A \lprec (B\lprec C)
        }{
          \vlin{\AXrule}{}{(A\lprec B)\lprec C \vdash (A\lprec B)\lprec C}{\vlhy{}}
        }{
          \vlin{\AXrule}{}{A \lprec (B\lprec C) \vdash A \lprec (B\lprec C)}{\vlhy{}}
        }
    }
  $$
  A similar derivation can be defined for $\sassl$ using the rule $\acutr$.

  For the right-to-left direction, we show that if the formula interpretation of both premises of the rule $\acutr$ are derivable in $\IBV$, then also the formula interpretation of the conclusion is provable in $\IBV$ (using $\sassr$).
  $$
  \od{
    \odo{
      \odi{
        \odh{
          \left(
            \Gamma^{\ltens} 
            \limp 
            (A\lprec (B \lprec C))
          \right)
          \ltens
          \odtP{
            \left(
              \odnP{(A\lprec B) \lprec C}{\sassr}{
                A\lprec (B\lprec C)
              }{}  
              \ltens 
              \Delta^{\ltens}
            \right) 
            \limp  D
          }{\curry}{
            (A\lprec (B\lprec C))  \limp (\Delta^{\ltens} \limp  D)
          }{}
        }
      }{\wlsr}{
        \left(\od{
          \odo{
            \odi{
              \odh{
                (
                  \Gamma^{\ltens} 
                  \limp 
                  (A\lprec (B\lprec C))
                )
                \limp 
                (A\lprec (B\lprec C))
              }
            }{\blsr}{
              \Gamma^{\ltens} 
              \ltens
              \odnP{(A\lprec (B\lprec C))\limp (A\lprec (B\lprec C))}{\baiur}{\lunit}{\brule}
            }{\brule}
          }{\tuur}{
            \Gamma^{\ltens} 
          }{}
        }\right)
        \limp 
        (\Delta^{\ltens} \limp  D)
      }{\wrule}
    }{\uncurry}{
      (\Gamma^{\ltens} \ltens \Delta^{\ltens}) \limp  D
    }{}
  }
  $$
  where $\Gamma^{\ltens}$ \resp{$\Delta^{\ltens}$} are formulas obtained by conjunction ($\ltens$) of all formulas in $\Gamma$ \resp{$\Delta$}.
  A similar argument applies for $\acutl$.
  \qed
\end{proof}

%%%%%%%%%%%%%%%%%%%%%%%%%%%%%%%%%%%%%%%%%%%%%%%%%%%%%%%%%%%%%%%%
%%%%%%%%%%%%%%%%%%%%%%%%%%%%%%%%%%%%%%%%%%%%%%%%%%%%%%%%%%%%%%%%
\section{Remarks on \Cref{sec:IBVtoBV}}
\label{app:IBVtoBV}
%%%%%%%%%%%%%%%%%%%%%%%%%%%%%%%%%%%%%%%%%%%%%%%%%%%%%%%%%%%%%%%%
%%%%%%%%%%%%%%%%%%%%%%%%%%%%%%%%%%%%%%%%%%%%%%%%%%%%%%%%%%%%%%%%

With the results of \Cref{sec:assoc}, we can give another proof of \Cref{thm:BVm}, which makes use of the sequent calculus instead of the deep inference system. The first observation is that a similar result as \Cref{thm:BiVextensNiML} also holds for $\NML$ and $\BV$. For this, let us define the \defn{system $\hBV$}, to be $\BVm$ (as defined in \Cref{fig:BVm}) where $\feq$ is replaced by $\hfeq$, which is defined via:
\begin{equation}
    \begin{array}{c@{\qquad}c}
      A\ltens (B\ltens C) \hfeq (A\ltens B)\ltens C
      &
      A\ltens B \hfeq B\ltens A
      \\ \\[-1ex]
      A\lpar (B\lpar C) \hfeq (A\lpar B)\lpar C
      &
      A\lpar B \hfeq B\lpar A 
    \end{array}
\end{equation}
i.e., $\feq$ with associativity for $\lseq$ removed. We now have the following:

%%%%%%%%%%%%%%%%%%%%%%%%%%%%%%%%%%%%%%%%%%%%%%%%%%%%%%%%%%%%%%%%
\begin{theorem}
  Let $A$ be a $\BVm$-formula. Then ~$\proves[\NMLm]A$~ iff ~$\proves[\hBV]A$~.
\end{theorem}
\begin{proof}
  The proof is similar to the proof of \Cref{thm:BiVextensNiML}, using one-sided sequents $\vdash \cneg A,B$ instead of $A\vdash B$ (see also~\cite{gug:str:01} and~\cite{acc:man:PN}).
  \qed
\end{proof}
%%%%%%%%%%%%%%%%%%%%%%%%%%%%%%%%%%%%%%%%%%%%%%%%%%%%%%%%%%%%%%%%

%%%%%%%%%%%%%%%%%%%%%%%%%%%%%%%%%%%%%%%%%%%%%%%%%%%%%%%%%%%%%%%%
\begin{theorem}
  Let $A$ be a $\BVm$-formula. Then ~$\proves[\BVm]A$~ iff ~$\proves[\hBV\cup\set{\sassr,\sassl}]A$~.
\end{theorem}
\begin{proof}
  This is trivial, as the rules $\sassr$, $\sassl$ can alway be used instead of the equality       $A\lseq (B\lseq C) \feq (A\lseq B)\lseq C$, and vice versa.
\qed
\end{proof}
%%%%%%%%%%%%%%%%%%%%%%%%%%%%%%%%%%%%%%%%%%%%%%%%%%%%%%%%%%%%%%%%

Note that the two rules $\sassr$ and $\sassl$ are applicable to $\IBV$-formulas and $\BVm$-formulas. With some abuse of notation, we can also allow them to be applied in a sequent calculus, as rewrite rules inside an arbitrary formula context. We can thus define two new proof systems:
\begin{equation}
  \begin{array}{rcl}
    \NMLp&~=~&\NMLm\cup\set{\sassr,\sassl}\\[1ex]
    \INMLp&~=~&\INML\cup\set{\sassr,\sassl}
  \end{array}
\end{equation}
%%%%%%%%%%%%%%%%%%%%%%%%%%%%%%%%%%%%%%%%%%%%%%%%%%%%%%%%%%%%%%%%
\begin{corollary}\label{cor:BVm}
    Let $A$ be a $\BVm$-formula. Then ~$\proves[\BVm]A$~ iff ~$\proves[\NMLp]A$~.
\end{corollary}
\begin{proof}
  Immediately from the previous two theorems.
  \qed
\end{proof}
%%%%%%%%%%%%%%%%%%%%%%%%%%%%%%%%%%%%%%%%%%%%%%%%%%%%%%%%%%%%%%%%

%%%%%%%%%%%%%%%%%%%%%%%%%%%%%%%%%%%%%%%%%%%%%%%%%%%%%%%%%%%%%%%%
\begin{corollary}\label{cor:IBV}
    Let $A$ be a $\IBV$-formula. Then ~$\proves[\IBV]A$~ iff ~$\proves[\INMLp]A$~.
\end{corollary}
\begin{proof}
  Immediately from \Cref{thm:BiVextensNiML} and the argument above.
  \qed
\end{proof}
%%%%%%%%%%%%%%%%%%%%%%%%%%%%%%%%%%%%%%%%%%%%%%%%%%%%%%%%%%%%%%%%

We now have the following:

%%%%%%%%%%%%%%%%%%%%%%%%%%%%%%%%%%%%%%%%%%%%%%%%%%%%%%%%%%%%%%%%
\begin{theorem}\label{thm:INMLstar}
  Let $A$ be a unit-free formula. Then ~$\proves[\INMLp]A$~ iff ~$\proves[\NMLp]\emb A$~.
\end{theorem}
\begin{proof}
  This proof is almost literally the same as the proof of \Cref{thm:NMLextensNiML}. We only need to observe that the two rules  $\sassr$ and $\sassl$ do not change the polarity of the formula they are applied to.
  \qed
\end{proof}
%%%%%%%%%%%%%%%%%%%%%%%%%%%%%%%%%%%%%%%%%%%%%%%%%%%%%%%%%%%%%%%%

We now immediately have our desired result:

%%%%%%%%%%%%%%%%%%%%%%%%%%%%%%%%%%%%%%%%%%%%%%%%%%%%%%%%%%%%%%%%
\thmBVm*
\begin{proof}
  For a unit-free formula $A$, we have ~$\proves[\IBV]A$~ iff ~$\proves[\IBVm]A$. Therefore, the theorem follows from \Cref{thm:INMLstar} and \Cref{cor:BVm} and \Cref{cor:IBV}.
\qed
\end{proof}
%%%%%%%%%%%%%%%%%%%%%%%%%%%%%%%%%%%%%%%%%%%%%%%%%%%%%%%%%%%%%%%%

%%%%%%%%%%%%%%%%%%%%%%%%%%%%%%%%%%%%%%%%%%%%%%%%%%%%%%%%%%%%%%%%
%%%%%%%%%%%%%%%%%%%%%%%%%%%%%%%%%%%%%%%%%%%%%%%%%%%%%%%%%%%%%%%%

%%%%%%%%%%%%%%%%%%%%%%%%%%%%%%%%%%%%%%%%%%%%%%%%%%%%%%%%%%%%%%%%
%%%%%%%%%%%%%%%%%%%%%%%%%%%%%%%%%%%%%%%%%%%%%%%%%%%%%%%%%%%%%%%%
%%%%%%%%%%%%%%%%%%%%%%%%%%%%%%%%%%%%%%%%%%%%%%%%%%%%%%%%%%%%%%%%
%%%%%%%%%%%%%%%%%%%%%%%%%%%%%%%%%%%%%%%%%%%%%%%%%%%%%%%%%%%%%%%%
\end{document}

%% file: macros.tex
\def\set#1{\{#1\}}

\def\tuple#1{\langle#1\rangle}
\usepackage{scalerel}

\renewcommand{\emptyset}{\varnothing}
\def\defn#1{{\itshape\bfseries\boldmath #1}}

\def\resp#1{(resp.~{#1})\xspace}%Respectively

\mathchardef\mhyphen="2D

\newcommand{\mydot}{\;.}

\newcommand{\quand}{\quad\mbox{and}\quad}

\def\IH{\mathsf{IH}}

\newcommand{\nodn}[4]{\vlinf{#2}{#4}{#3}{#1}}
\newcommand{\nodt}[4]{\vlidf{#2}{#4}{#3}{#1}}

\newcommand{\oodh}[1]{\od{\odh{#1}}}
\newcommand{\vlhtrl}[5]{\vltrl{#1}{#2}{#3}{#4}{\vlhy{~}}{\vlhy{#5}}{\vlhy{~}}}

\newcommand{\atomset}{\mathcal A}

\def\feq{\equiv}
\def\hfeq{\hat\equiv}

\newcommand{\ctx}[1][\chole]{[#1]}
\newcommand{\Ctx}[1][]{\left[{#1}\right]}

\newcommand{\mathca}{}
\newcommand{\cC}[1][\chole]{\mathca C\!\Ctx[#1]}
\newcommand{\cP}[1][\chole]{\mathca P\!\Ctx[#1]}
\newcommand{\cN}[1][\chole]{\mathca N\!\Ctx[#1]}

\newcommand{\cPp}[1][\chole]{\mathca P'\Ctx[#1]}

\newcommand{\proves}[1][]{\mathord{\vdash_{#1}\,}}
\def\dD{\mathcal D}
\def\dDp{\mathcal{P}}
\def\dDn{\mathcal{N}}
\def\dDq{\mathcal{Q}}
\def\dDm{\mathcal{M}}

\def\MLL{\mathsf{MLL}}

\def\MLLx{\mathsf{MLL}_{\mixr}}
\def\BV{\mathsf{BV}}
\def\hBV{\widehat{\mathsf{BV}}}
\def\BVm{\mathsf{BV^-}}
\def\MiLL{\mathsf{IMLL}}
\def\IMLL{\mathsf{IMLL}}

\def\BiV{\mathsf{IBV}}
\def\IBV{\mathsf{IBV}}
\def\IBVm{\mathsf{IBV^-}}
\def\SBiV{\mathsf{SIBV}}
\def\SIBV{\mathsf{SIBV}}
\def\pomset{\mathsf{pomset}}
\def\sS{\mathsf{S}}

\def\XS{\mathsf{X}}

\def\NML{\mathsf{NML}}
\def\NMLm{\mathsf{NML^-}}
\def\NMLp{\mathsf{NML^\star}}

\def\NiML{\mathsf{INML}}
\def\INML{\mathsf{INML}}

\def\INMLp{\mathsf{INML^\star}}

\def\LJ{\mathsf{LJ}}

\def\NLJ{\mathsf{NLJ}}
\def\BLJ{\mathsf{BLJ}}

\defedgetype{DOT}{densely dotted}{}
\defedgetype{DDOT}{>=stealth,->,draw,densely dotted}{}

\def\axrule{\mathsf {ax}}
\def\AXrule{\mathsf {AX}}
\def\cutr{\mathsf {cut}}
\def\mixr{\mathsf{mix}}
\def\airule{\mathsf{ai}}
\def\irule{\mathsf{i}}
\def\swir{\mathsf{s}}
\def\sqir{\mathsf{sq}}
\def\qrule{\mathsf{q}}
\def\rrule{\mathsf{r}}
\def\urule{\mathsf{u}}

\def\comrule{\mathsf{com}}

\def\refrule{\mathsf{ref}}
\def\curryrule{\mathsf{cur}}
\def\defr{\mathsf{def}}

\def\curry{\mathsf{cur}}
\def\uncurry{\mathsf{ruc}}

\def\axrule{\mathsf{ax}}
\def\llimp{\mathsf{\limp_L}}
\def\rlimp{\mathsf{\limp_R}}
\def\lltens{\mathsf{\ltens_L}}
\def\rltens{\mathsf{\ltens_R}}
\def\lunitr{\lunit_R}
\def\lunitl{\lunit_L}

\def\idr{\irule\mathord{\downarrow}}

\def\aidr{\airule\mathord{\downarrow}}
\def\aiur{\airule\mathord{\uparrow}}
\def\qdr{\qrule\mathord{\downarrow}}

\def\rur{\rrule\mathord{\uparrow}}
\def\udr{\urule\mathord{\downarrow}}

\def\squr{\mathsf{sq_R}}
\def\squl{\mathsf{sq_L}}

\def\waidr{\mathsf{ai}_\downarrow^\circ}
\def\baidr{\mathsf{ai}_\downarrow^\bullet}
\def\baiur{\mathsf{ai}_\uparrow^\bullet}
\def\widr{\mathsf{i}_\downarrow^\circ}
\def\biur{\mathsf{i}_\uparrow^\bullet}

\def\tudr{\mathsf{u}_\downarrow^{\ltens}}
\def\ludr{\mathsf{u}_\downarrow^{\limp}}
\def\tuur{\mathsf{u}_\uparrow^{\ltens}}
\def\luur{\mathsf{u}_\uparrow^{\limp}}

\def\wsudr{\mathsf{u}_\downarrow^\lseq}
\def\wcudr{\mathsf{u}_\downarrow^\lcoseq}
\def\bsudr{\mathsf{u}_\uparrow^\lseq}
\def\bcudr{\mathsf{u}_\uparrow^\lcoseq}

\def\tassr{\mathsf{asso}^{\ltens}}
\def\sassr{\mathsf{asso_R^{\lseq}}}
\def\sassl{\mathsf{asso_L^{\lseq}}}

\def\wlsr{\mathsf{s_L^\circ}}
\def\wtsr{\mathsf{s_R^\circ}}
\def\blsr{\mathsf{s_L^\bullet}}
\def\btsr{\mathsf{s_R^\bullet}}
\def\wlsq{\mathsf{sq_R^\circ}}
\def\wrsq{\mathsf{sq_L^\circ}}
\def\blsq{\mathsf{sq_R^\bullet}}
\def\brsq{\mathsf{sq_L^\bullet}}
\def\wrr{\refrule^\circ}
\def\brr{\refrule^\bullet}

\def\wqdr{\mathsf{q}_\downarrow^\circ}
\def\bqdr{\mathsf{q}_\downarrow^\bullet}
\def\wqur{\mathsf{q}_\uparrow^\circ}
\def\bqur{\mathsf{q}_\uparrow^\bullet}

\def\tcomr{\comrule^{\ltens}}